\documentclass[a4paper,USenglish]{lipics-v2018}
\usepackage{xspace}
\usepackage{microtype}\usepackage{stmaryrd}
\usepackage{bbm}
\usepackage{amsmath}
\usepackage{tikz}
\usetikzlibrary{decorations.pathreplacing,hobby,backgrounds,trees,arrows,automata,shapes,positioning,fit,patterns,calc,matrix}

\tikzstyle{snode}=[circle,draw=black!70,thick,align=center,minimum size = 0.5cm,inner sep =0pt]
\tikzstyle{bobox}=[rectangle,rounded corners=5pt,draw=black!70,thick,fill=orange!20]
\tikzstyle{triangle} =[regular polygon, regular polygon sides=3,rounded corners=5pt,draw=black!70,fill=gray!20]
\tikzstyle{lines} =[draw=black!70,thick]

\nolinenumbers

\newcommand{\ptime}{\ensuremath{\mathsf{P}}\xspace}
\newcommand{\pspace}{\ensuremath{\mathsf{PSpace}}\xspace}
\newcommand{\exptime}{\ensuremath{\mathsf{EXPTime}}\xspace}
\newcommand{\logspace}{\ensuremath{\mathsf{LogSpace}}\xspace}
\newcommand{\nlog}{\ensuremath{\mathsf{NL}}\xspace}
\newcommand{\nfa}{\ensuremath{\mathsf{NFA}}\xspace}
\newcommand{\nfas}{\ensuremath{\mathsf{NFAs}}\xspace}
\newcommand{\dfa}{\ensuremath{\mathsf{DFA}}\xspace}
\newcommand{\dfas}{\ensuremath{\mathsf{DFAs}}\xspace}

\newcommand{\nat}{\ensuremath{\mathbb{N}}\xspace}
\newcommand{\integ}{\ensuremath{\mathbb{Z}}\xspace}

\newcommand{\quozk}{\ensuremath{{\integ}/{k\integ}}\xspace}

\newcommand\contentmorphism{\ensuremath{\textsf{alph}}}
\newcommand\cont[1]{\ensuremath{\contentmorphism(#1)}}

\newcommand{\reg}{\ensuremath{\textup{REG}}\xspace}
\newcommand{\mso}{\ensuremath{\textup{MSO}}\xspace}

\newcommand{\at}{\ensuremath{\textup{AT}}\xspace}

\newcommand{\sfr}{\ensuremath{\textup{SF}}\xspace}

\newcommand{\sttp}[1]{\ensuremath{\textup{ST}[#1]}\xspace}

\newcommand{\stzer}{\sttp{0}}
\newcommand{\sthone}{\sttp{{1}/{2}}}
\newcommand{\stone}{\sttp{1}}
\newcommand{\sthtwo}{\sttp{{3}/{2}}}
\newcommand{\sttwo}{\sttp{2}}
\newcommand{\sththree}{\sttp{{5}/{2}}}

\newcommand{\bool}[1]{\ensuremath{\mathit{Bool}(#1)}\xspace}
\newcommand{\pol}[1]{\ensuremath{\mathit{Pol}(#1)}\xspace}
\newcommand{\bpol}[1]{\ensuremath{\mathit{BPol}(#1)}\xspace}

\newcommand{\polp}[2]{\ensuremath{\mathit{Pol}_{#2}(#1)}\xspace}
\newcommand{\polk}[1]{\polp{#1}{k}}

\newcommand{\polrelp}[1]{\ensuremath{\leqslant_{#1}}\xspace}
\newcommand{\polrelk}{\polrelp{k}}

\newcommand{\typ}[2]{\ensuremath{[#1]_{#2}}\xspace}
\newcommand{\ctype}[1]{\typ{#1}{\Cs}}
\newcommand{\cmult}{\ensuremath{\mathbin{\scriptscriptstyle\bullet}}}

\usepackage{stmaryrd}

\newcommand{\As}{\ensuremath{\mathcal{A}}\xspace}

\newcommand{\Cs}{\ensuremath{\mathcal{C}}\xspace}
\newcommand{\Ds}{\ensuremath{\mathcal{D}}\xspace}

\newcommand{\Js}{\ensuremath{\mathcal{J}}\xspace}

\newcommand{\Lb}{\ensuremath{\mathbf{L}}\xspace}

\newcommand{\frE}{\ensuremath{\mathbbm{E}}\xspace}

\newcommand{\frP}{\ensuremath{\mathbbm{P}}\xspace}
\newcommand{\frQ}{\ensuremath{\mathbbm{Q}}\xspace}

\newcommand{\frS}{\ensuremath{\mathbbm{S}}\xspace}
\newcommand{\frT}{\ensuremath{\mathbbm{T}}\xspace}
\newcommand{\frU}{\ensuremath{\mathbbm{U}}\xspace}
\newcommand{\frV}{\ensuremath{\mathbbm{V}}\xspace}

\newcommand{\vari}{quotienting Boolean algebra\xspace}

\newcommand{\varis}{quotienting Boolean algebras\xspace}

\newcommand{\pvari}{quotienting lattice\xspace}

\newcommand{\pvaris}{quotienting lattices\xspace}

\newcommand{\varie}{variety\xspace}

\newcommand{\varies}{varieties\xspace}

\newcommand{\pvarie}{positive variety\xspace}

\newcommand{\pvaries}{positive varieties\xspace}
\newcommand{\pVaries}{Positive varieties\xspace}

\def\inv{^{-1}}

\newcounter{sauvegarde}
\newcommand\adjustc[1]{
  \setcounter{sauvegarde}{\thetheorem}
  \setcounterref{theorem}{#1}
  \addtocounter{theorem}{-1}
}

\newcommand\restorec{
  \setcounter{theorem}{\thesauvegarde}
}

\title{The complexity of separation for levels in concatenation hierarchies}

\titlerunning{The complexity of separation for levels in concatenation hierarchies}

\author{Thomas Place}{LaBRI, Bordeaux University and IUF, France}{}{}{}
\author{Marc Zeitoun}{LaBRI, Bordeaux University, France}{}{}{}
\authorrunning{T. Place and M. Zeitoun}

\Copyright{Thomas Place and Marc Zeitoun}

\subjclass{}
\funding{Both authors acknowledge support from the DeLTA project (ANR-16-CE40-0007).}

\EventEditors{Sumit Ganguly and Paritosh Pandya}
\EventNoEds{2}
\EventLongTitle{38th IARCS Annual Conference on Foundations of Software Technology and Theoretical Computer Science (FSTTCS 2018)}
\EventShortTitle{FSTTCS 2018}
\EventAcronym{FSTTCS}
\EventYear{2018}
\EventDate{December 11--13, 2018}
\EventLocation{Ahmedabad, India}
\EventLogo{}
\SeriesVolume{122}
\ArticleNo{47}

\theoremstyle{plain}
\newtheorem{proposition}[theorem]{Proposition}
\newtheorem{fact}[theorem]{Fact}

\begin{document}

\maketitle

\begin{abstract}
  We investigate the complexity of the separation problem associated to classes of regular languages. For a class \Cs, \Cs-separation takes two regular languages as input and asks whether there exists a third language in \Cs which includes the first and is disjoint from the second. First, in contrast with the situation for the classical membership problem, we prove that for most classes \Cs, the complexity of \Cs-separation does not depend on how the input languages are represented: it is the same for nondeterministic finite automata and monoid morphisms. Then, we investigate specific classes belonging to finitely based concatenation hierarchies. It was recently proved that the problem is always decidable for levels 1/2 and 1 of any such hierarchy (with inefficient algorithms). Here, we build on these results to show that when the alphabet is fixed, there are polynomial time algorithms for both levels. Finally, we investigate levels 3/2 and 2 of the famous Straubing-Thérien hierarchy. We show that separation is \pspace-complete for level 3/2 and between \pspace-hard and \exptime for level 2.
\end{abstract}

\section{Introduction}
\label{sec:intro}
For more than 50 years, a significant research effort in theoretical computer science was made to solve the membership problem for regular languages. This problem consists in determining whether a class of regular languages is decidable, that is, whether there is an algorithm inputing a regular language and outputing `yes' if the language belongs to the investigated class, and `no' otherwise.

Many results were obtained in a long and fruitful line of research. The most prominent one is certainly Schützenberger's theorem~\cite{sfo}, which gives such an algorithm for the class of star-free languages. For most interesting classes also, we know precisely the computational cost of the membership problem. As can be expected, this cost depends on the way the input language is given. Indeed, there are several ways to input a regular language. For instance, it can be given by a nondeterministic finite automaton (\nfa), or, alternately, by a morphism into a finite monoid. While obtaining an \nfa representation from a morphism into a monoid has only a linear cost, the converse direction is much more expensive: from an \nfa with $n$ states, the smallest monoid recognizing the same language may have an exponential number of elements (the standard construction yields $2^{n^2}$ elements). This explains why the complexity of the membership problem depends on the representation of the input. For instance, for the class of star-free languages, it is \pspace-complete if one starts from \nfas (and actually, even from \dfas~\cite{chofo}) while it is \nlog when starting from monoid morphisms.

Recently, another problem, called separation, has replaced membership as the cornerstone in the investigation of regular languages. It takes as input \emph{two} regular langages instead of one, and asks whether there exists a third language from the class under investigation including the first input language and having empty intersection with the second one. This problem has served recently as a major ingredient in the resolution of difficult membership problems, such as the so-called dot-depth two problem~\cite{pz:qalt:2014} which remained open for 40 years (see~\cite{pztale,PZ:generic_csr_tocs:18,jep-dd45} for recent surveys on the topic). Dot-depth two is a class belonging to a famous \emph{concatenation hierarchy} which stratifies the star-free languages: the dot-depth~\cite{BrzoDot}. A specific concatenation hierarchy is built in a generic way. One starts from a base class (level 0 of the hierarchy) and builds increasingly growing classes (called levels and denoted by 1/2, 1, 3/2, 2, $\dots$) by alternating two standard closure operations: polynomial and Boolean closure. Concatenation hierarchies account for a significant part of the open questions in this research area. The state of the art regarding separation is captured by only three results~\cite{pzbpol,pbp}: in finitely based concatenation hierarchies (i.e. those whose basis is a finite class) levels 1/2, 1 and 3/2 have decidable separation.  Moreover, using specific transfer results~\cite{pzsuccfull}, this can be pushed to the levels 3/2 and 2 for the two most famous finitely based hierarchies: the dot-depth~\cite{BrzoDot} and the Straubing-Thérien hierarchy~\cite{StrauConcat,TheConcat}.

Unlike the situation for membership and despite these recent decidability results for separability in concatenation hierarchies, the complexity of the problems and of the corresponding algorithms has not been investigated so far (except for the class of piecewise testable languages~\cite{martens,pvzmfcs13,Masopust18}, which is level 1 in the Straubing-Thérien hierarchy). The aim of this paper is to establish such complexity results. Our contributions are the following:
\begin{itemize}
\item We present a \textbf{generic} reduction, which shows that for many natural classes, the way the input is given (by \nfas or finite monoids) has \textbf{no impact} on the complexity of the separation problem. This is proved using two \logspace reductions from one problem to the other. This situation is surprising and opposite to that of the membership problem, where an exponential blow-up is unavoidable when going from \nfas to monoids.
\item Building on the results of~\cite{pzbpol}, we show that when the alphabet is fixed, there are polynomial time algorithms for levels 1/2 and 1 in any finitely based hierarchy. 	\item We investigate levels 3/2 and 2 of the famous Straubing-Thérien hierarchy, and we show that separation is \pspace-complete for level 3/2 and between \pspace-hard and \exptime for level 2. The upper bounds are based on the results of~\cite{pzbpol} while the lower bounds are based on independent reductions.
\end{itemize}

\noindent
{\bf Organization.} In Section~\ref{sec:prelims}, we give preliminary terminology on the objects investigated in the paper. Sections~\ref{sec:nfatomono}, \ref{sec:fixalph} and~\ref{sec:classic} are then devoted to the three above points. Due to space limitations, many proofs are postponed to the appendix.

\section{Preliminaries}
\label{sec:prelims}
In this section, we present the key objects of this paper. We define words and regular languages, classes of languages, the separation problem and finally, concatenation hierarchies.

\subsection{Words and regular languages}

An alphabet is a \emph{finite} set $A$ of symbols, called \emph{letters}. Given some alphabet $A$, we denote by $A^+$ the set of all nonempty finite words and by $A^{*}$ the set of all finite words over $A$ (\emph{i.e.}, $A^* = A^+ \cup \{\varepsilon\}$). If $u \in A^*$ and $v \in A^*$ we write $u \cdot v \in A^*$ or $uv \in A^*$ for the concatenation of $u$ and~$v$. A \emph{language} over an alphabet $A$ is a subset of $A^*$. Abusing terminology, if $u \in A^*$ is some word, we denote by $u$ the singleton language~$\{u\}$. It is standard to extend concatenation to languages: given $K,L \subseteq A^*$, we write~$KL = \{uv \mid u \in K \text{ and } v \in L\}$. Moreover, we also consider marked concatenation, which is less standard. Given $K,L \subseteq A^*$, \emph{a marked concatenation} of~$K$ with $L$ is a language of the form $KaL$, for some $a \in A$.

We consider \emph{regular languages}, which can be equivalently defined by \emph{regular expressions}, \emph{nondeterministic finite automata}~(\nfas), \emph{finite monoids} or \emph{monadic second-order logic} (\mso). In the paper, we investigate the separation problem which takes regular languages as input. Since we are focused on complexity, how we represent these languages in our inputs matters. We shall consider two kinds of representations: \nfas and monoids. Let us briefly recall these objects and fix the terminology (we refer the reader to~\cite{pingoodref} for details).

\medskip

\noindent

{\bf NFAs.} An \nfa is a tuple $\As = (A,Q,\delta,I,F)$ where $A$ is an alphabet, $Q$ a finite set of states, $\delta \subseteq Q \times A \times Q$ a set of transitions, $I \subseteq Q$ a set of initial states and $F \subseteq Q$ a set of final states.  The language $L(\As) \subseteq A^*$ consists of all words labeling a run from an initial state to a final state. The regular languages are exactly those which are recognized by an \nfa. Finally, we write ``\dfa'' for \emph{deterministic} finite automata, which are defined in the standard way.

\medskip

\noindent

{\bf Monoids.} We turn to the algebraic definition of regular languages. A \emph{monoid} is a set $M$ endowed with an associative multiplication $(s,t) \mapsto s\cdot t$ (also denoted by~$st$) having a neutral element $1_M$, \emph{i.e.}, such that ${1_M}\cdot s=s\cdot {1_M}=s$ for every $s \in M$. An \emph{idempotent} of a monoid $M$ is an element $e \in M$ such that $ee = e$.

Observe that $A^{*}$ is a monoid whose multiplication is concatenation (the neutral element is $\varepsilon$). Thus, we may consider monoid morphisms $\alpha: A^* \to M$ where $M$ is an arbitrary monoid. Given such a morphism, we say that a language $L\subseteq A^*$ is \emph{recognized} by~$\alpha$ when there exists a set $F \subseteq M$ such that $L = \alpha\inv(F)$. It is well-known that the regular languages are also those which are recognized by a morphism into a \emph{finite} monoid. When representing a regular language $L$ by a morphism into a finite monoid, one needs to give both the morphism $\alpha: A^* \to M$ (\emph{i.e.}, the image of each letter) and the set $F \subseteq M$ such that $L = \alpha\inv(F)$.

\subsection{Classes of languages and separation}

A class of languages \Cs is a correspondence $A \mapsto \Cs(A)$ which, to an alphabet $A$, associates a set of languages $\Cs(A)$ over $A$.

\begin{remark}

  When two alphabets $A,B$ satisfy $A \subseteq B$, the definition of classes does not require $\Cs(A)$ and $\Cs(B)$ to be comparable. In fact, it may happen that a particular language $L \subseteq A^* \subseteq B^*$ belongs to $\Cs(A)$ but not to $\Cs(B)$ (or the opposite).	For example, we may consider the class \Cs defined by $\Cs(A) =  \{\emptyset,A^*\}$ for every alphabet $A$. When $A \subsetneq B$, we have $A^* \in \Cs(A)$ while $A^* \not\in \Cs(B)$.

\end{remark}

We say that \Cs is a \emph{lattice} when for every alphabet $A$, we have $\emptyset,A^* \in \Cs(A)$ and $\Cs(A)$ is closed under finite union and finite intersection: for any $K,L \in \Cs(A)$, we have $K \cup L \in \Cs(A)$ and $K \cap L \in \Cs(A)$. Moreover, a \emph{Boolean algebra} is a lattice \Cs which is additionally closed under complement: for any $L \in \Cs(A)$, we have $A^* \setminus L \in \Cs(A)$. Finally, a class \Cs is \emph{quotienting} if it is closed under quotients. That is, for every alphabet $A$, $L \in \Cs(A)$ and word $u \in A^*$, the following properties~hold:
\[
  u^{-1}L \stackrel{\text{def}}{=}\{w\in A^*\mid uw\in L\} \text{\quad and\quad} Lu^{-1} \stackrel{\text{def}}{=}\{w\in A^*\mid wu\in L\}\text{\quad both belong to $\Cs(A)$}.
\]
All classes that we consider in the paper are (at least) \pvaris consisting of \emph{regular languages}. Moreover, some of them satisfy an additional property called \emph{closure under inverse image}.

Recall that $A^*$ is a monoid for any alphabet $A$. We say that a class \Cs is \emph{closed under inverse image} if for every two alphabets $A,B$, every monoid morphism $\alpha: A^* \to B^*$ and every language $L \in \Cs(B)$, we have $\alpha\inv (L) \in \Cs(A)$. A \pvari (resp. \vari) closed under inverse image is called a \emph{\pvarie} (resp. \emph{\varie}).

\medskip

\noindent

{\bf Separation.} Consider a class of languages \Cs. Given an alphabet $A$ and two languages $L_1,L_2 \subseteq A^*$, we say that $L_1$ is \Cs-separable from $L_2$ when there exists a third language $K \in \Cs(A)$ such that $L_1 \subseteq K$ and $L_2 \cap K = \emptyset$. In particular, $K$ is called a \emph{separator} in \Cs. The \Cs-separation problem is now defined as follows:

\begin{tabular}{ll}

  {\bf Input:} & An alphabet $A$ and two regular languages $L_1,L_2 \subseteq A^*$. \\

  {\bf Output:} & Is $L_1$ \Cs-separable from $L_2$ ?

\end{tabular}

\begin{remark}

  Separation generalizes the simpler \emph{membership problem}, which asks whether a single regular language belongs to \Cs. Indeed $L \in \Cs$ if and only if $L$ is \Cs-separable from $A^* \setminus L$ (which is also regular and computable from $L$).

\end{remark}

Most papers on separation are mainly concerned about decidability. Hence, they do not go beyond the above presentation of the problem  (see~\cite{martens,pz:qalt:2014,pzfo,pzbpol} for example). However, this paper specifically investigates complexity. Consequently, we shall need to be more precise and take additional parameters into account. First, it will be important to specify whether the alphabet over which the input languages is part of the input (as above) or a constant. When considering separation for some fixed alphabet $A$, we shall speak of ``$\Cs(A)$-separation''. When the alphabet is part of the input, we simply speak of ``\Cs-separation''.

Another important parameter is how the two input languages are represented. We shall consider \nfas and monoids. We speak of \emph{separation for \nfas and separation for monoids}. Note that one may efficiently reduce the latter to the former. Indeed, given a language $L \subseteq A^*$ recognized by some morphism $\alpha: A^* \to M$, it is simple to efficiently compute a \nfa with $|M|$ states recognizing $L$ (see~\cite{pingoodref} for example). Hence, we have the following lemma.

\begin{lemma} \label{lem:easyreduc}

  For any class \Cs, there is a \logspace reduction from \Cs-separation for monoids to \Cs-separation for \nfas.

\end{lemma}

Getting an efficient reduction for the converse direction is much more difficult since going from \nfas (or even \dfas) to monoids usually involves an exponential blow-up. However, we shall see in Section~\ref{sec:nfatomono} that for many natural classes \Cs, this is actually possible.

\subsection{Concatenation hierarchies}

We now briefly recall the definition of concatenation hierarchies. We refer the reader to~\cite{PZ:generic_csr_tocs:18} for a more detailed presentation. A particular concatenation hierarchy is built from a starting class of languages \Cs, which is called its \emph{basis}. In order to get robust properties, we restrict~\Cs to be a \vari of regular languages. The basis is the only parameter in the construction. Once fixed, the construction is generic: each new level is built from the previous one by applying generic operators: either Boolean closure, or polynomial closure. Let us first define these two operators.

\medskip

\noindent

{\bf Definition.} Consider a class \Cs. We denote by \bool{\Cs} the \emph{Boolean closure} of \Cs: for every alphabet $A$, $\bool{\Cs}(A)$ is the least set containing $\Cs(A)$ and closed under Boolean operations. Moreover, we denote by \pol{\Cs} the \emph{polynomial closure} of \Cs: for every alphabet $A$, $\pol{\Cs}(A)$ is the least set containing $\Cs(A)$ and closed under union and marked concatenation (if $K,L \in \pol{\Cs}(A)$ and $a \in A$, then $K \cup L,KaL \in \pol{\Cs}(A)$).

Consider a \vari of regular languages \Cs. The concatenation hierarchy of basis \Cs is defined as follows. Languages are classified into levels of two kinds: full levels (denoted by 0, 1, 2,$\dots$) and half levels (denoted by 1/2, 3/2, 5/2,$\dots$). Level $0$ is the basis (\emph{i.e.}, \Cs) and for every $n \in \nat$,

\begin{itemize}

\item The \emph{half level} $n+1/2$ is the \emph{polynomial closure} of the previous full level, \emph{i.e.}, of level $n$.

\item The \emph{full level} $n+1$ is the \emph{Boolean closure} of the previous half level, \emph{i.e.}, of level $n+1/2$.

\end{itemize}

\begin{center}

  \begin{tikzpicture}[scale=.9]

    \node[anchor=east] (l00) at (0.0,0.0) {{\large $0$}};

    \node[anchor=east] (l12) at (2.0,0.0) {\large $1/2$};

    \node[anchor=east] (l11) at (4.0,0.0) {\large $1$};

    \node[anchor=east] (l32) at (6.0,0.0) {\large $3/2$};

    \node[anchor=east] (l22) at (8.0,0.0) {\large $2$};

    \node[anchor=east] (l52) at (10.0,0.0) {\large $5/2$};

    \draw[very thick,->] (l00) to node[above] {$Pol$} (l12);

    \draw[very thick,->] (l12) to node[below] {$Bool$} (l11);

    \draw[very thick,->] (l11) to node[above] {$Pol$} (l32);

    \draw[very thick,->] (l32) to node[below] {$Bool$} (l22);

    \draw[very thick,->] (l22) to node[above] {$Pol$} (l52);

    \draw[very thick,dotted] (l52) to ($(l52)+(1.0,0.0)$);

  \end{tikzpicture}

\end{center}

We write $\frac 12 \nat = \{0,1/2,1,2,3/2,3,\dots\}$ for the set of all possible levels in a concatenation hierarchy. Moreover, for any basis \Cs and $n \in \frac12 \nat$, we write $\Cs[n]$ for level $n$ in the concatenation hierarchy of basis \Cs. It is known that every half-level is a \pvari and every full level is a \vari (see~\cite{PZ:generic_csr_tocs:18} for a recent proof).

We are interested in finitely based concatenation hierarchies: if \Cs is the basis, then $\Cs(A)$ is finite for every alphabet $A$. Indeed, it was shown in~\cite{pzbpol} that for such hierarchies separation is always decidable for the levels 1/2 and 1  (in fact, while we do not discuss this in the paper, this is also true for level 3/2, see~\cite{pbp} for a preliminary version). In Section~\ref{sec:fixalph}, we build on the results  of~\cite{pzbpol} and show that when the alphabet is fixed, this can be achieved in polynomial time for both levels 1/2 and 1. Moreover, we shall also investigate the famous \emph{Straubing-Th\'erien} hierarchy in Section~\ref{sec:classic}. Our motivation for investigating this hierarchy in particular is that the results of~\cite{pzbpol} can be pushed to levels 3/2 and 2 in this special case.

\section{Handling \nfas}
\label{sec:nfatomono}
In this section, we investigate how the representation of input languages impact the complexity of separation. We prove that for many natural classes \Cs (including most of those considered in the paper), \Cs-separation has the same complexity for \nfas as for monoids. Because of these results, we shall be able to restrict ourselves to monoids in later sections.

\begin{remark}
  This result highlights a striking difference between separation and the simpler membership problem. For most classes \Cs, \Cs-membership is strictly harder for \nfas than for monoids. This is because when starting from a \nfa, typical membership algorithms require to either determinize \As or compute a monoid morphism recognizing $L(\As)$  which involves an exponential blow-up in both cases. Our results show that the situation differs for separation.
\end{remark}

We already have a generic efficient reduction from \Cs-separation for monoids to \Cs-separation for \nfas (see Lemma~\ref{lem:easyreduc}). Here, we investigate the opposite direction: given some class \Cs, is it possible to \emph{efficiently} reduce \Cs-separation for \nfas to \Cs-separation for monoids ? As far as we know, there exists no such reduction which is generic to all classes \Cs.

\begin{remark}
  There exists an \emph{inefficient} generic reduction from separation for \nfas to the separation for monoids. Given as input two \nfas $\As_1,\As_2$, one may compute monoid morphisms recognizing $L(\As_1)$ and $L(\As_2)$. This approach is not satisfying as it involves an exponential blow-up: we end-up with monoids $M_i$ of size $2^{|Q_i|^2}$ where $Q_i$ is the set of states of $\As_i$.
\end{remark}

Here, we present a set of conditions applying to a pair of classes $(\Cs,\Ds)$. When they are satisfied, there exists an efficient reduction from \Cs-separation for \nfas to \Ds-separation for monoids. By themselves, these conditions are abstract. However, we highlight two concrete applications. First, for every \pvarie \Cs, the pair $(\Cs,\Cs)$ satisfies the conditions. Second, for every finitely based concatenation hierarchies of basis \Cs, there exists another finite basis \Ds such that for every $n \in \frac12 \nat$, the pair $(\Cs[n],\Ds[n])$ satisfies the conditions

We first introduce the notions we need to present the reduction and the conditions required to apply it. Then, we state the reduction itself and its applications.

\subsection{Generic theorem}

We fix a special two letter alphabet $\frE = \{0,1\}$. For the sake of improved readability, we abuse terminology and assume that when considering an arbitrary alphabet $A$, it always has empty intersection with \frE. This is harmless as we may work up to bijective renaming.

We exhibit conditions applying to a pair of classes $(\Cs,\Ds)$. Then, we prove that they imply the existence of an efficient reduction from \Cs-separation for \nfas to \Ds-separation for monoids. This reduction is based on a construction which takes as input a \nfa \As (over some arbitrary alphabet $A$) and builds a modified version of the language $L(\As)$ (over $A \cup \frE$) which is recognized by a ``small'' monoid. Our conditions involve two kinds of hypotheses:
\begin{enumerate}
\item First, we need properties related to inverse image: ``\Ds must be an an extension of \Cs''.
\item The construction is parametrized by an object called ``tagging''. We need an algorithm which builds special taggings (with respect to \Ds) efficiently.
\end{enumerate}
We now make these two notions more precise. Let us start with extension.

\medskip
\noindent
{\bf Extensions.} Consider two classes \Cs and \Ds. We say that \Ds is an extension of \Cs when for every alphabet $A$, the two following conditions hold:
\begin{itemize}
\item If $\gamma: (A \cup \frE)^* \to A^*$ is the morphism defined by $\gamma(a) = a$ for $a \in A$ and $\gamma(b) = \varepsilon$ for $b \in \frE$, then for every $K \in \Cs(A)$, we have $\gamma\inv(K)  \in \Ds(A \cup \frE)$.
\item For every $u \in \frE^*$, if $\lambda_u: A^* \to (A \cup \frE)^*$ is the morphism defined by $\lambda_u(a) = au$ for $a \in A$, then for every $K \in \Ds(A \cup \frE)$, we have $\lambda_u\inv(K) \in \Cs(A)$.
\end{itemize}
\pVaries give an important example of extension. Since they are closed under inverse image, it is immediate that for every \pvarie \Cs, \Cs is an extension of itself.

\medskip
\noindent
{\bf Taggings.} A \emph{tagging} is a pair $P = (\tau: \frE^* \to T,G)$ where $\tau$ is a morphism into a finite monoid and $G \subseteq T$. We call $|G|$ the \emph{rank} of $P$ and $|T|$ its size. Moreover, given some \nfa $\As = (A,Q,\delta,I,F)$, $P$ is \emph{compatible with \As} when the rank $|G|$ is larger than $|\delta|$.

For our reduction, we shall require special taggings. Consider a class \Ds and a tagging $P = (\tau: \frE^* \to T,G)$. We say that $P$ \emph{fools} \Ds when, for every alphabet $A$ and every morphism $\alpha: (A \cup \frE)^* \to M$ into a finite monoid $M$, if all languages recognized by $\alpha$ belong to $\bool{\Ds}(A \cup \frE)$, then, there exists $s \in M$, such that for every $t \in G$, we have $w_t \in \frE^*$ which satisfies $\alpha(w_t) = s$ and $\tau(w_t) = t$.

Our reduction requires an efficient algorithm for computing taggings which fool the output class \Ds. Specifically, we say that a class \Ds is \emph{smooth} when, given as input $k \in \nat$, one may compute in \logspace (with respect to $k$) a tagging of rank at least $k$ which fools \Ds.

\medskip
\noindent
{\bf Main theorem.} We may now state our generic reduction theorem. The statement has two variants depending on whether the alphabet is fixed or not.

\begin{theorem} \label{thm:autoreduc}
  Let $\Cs,\Ds$ be \pvaris such that \Ds is smooth and extends \Cs. Then the two following properties hold:
  \begin{itemize}
  \item There is a \logspace reduction from \Cs-separation for \nfas to \Ds-separation for monoids.
  \item For every fixed alphabet $A$, there is a \logspace reduction from $\Cs(A)$-separation for \nfas to $\Ds(A \cup \frE)$-separation for monoids.
  \end{itemize}
\end{theorem}

We have two main applications of Theorem~\ref{thm:autoreduc} which we present at the end of the section. Let us first describe the reduction. As we explained, we use a construction building a language recognized by a ``small'' monoid out of an input \nfa and a compatible tagging.

\medskip

Consider a \nfa $\As = (A,Q,\delta,I,F)$ and let $P = (\tau: \frE^* \to T,G)$ be a compatible tagging (i.e. $|\delta| \leq |G|$). We associate a new language $L[\As,P]$ over the alphabet $A \cup \frE$ and show that one may efficiently compute a recognizing monoid whose size is polynomial with respect to $|Q|$ and the rank of $P$ (i.e $|G|$). The construction involves two steps. We first define an intermediary language $K[\As,P]$ over the alphabet $A \times T$ and then define $L[\As,P]$ from it.

We define $K[\As,P] \subseteq (A \times T)^*$ as the language recognized by a new \nfa $\As[P]$ which is built by relabeling the transitions of \As.  Note that the definition of $\As[P]$ depends on arbitrary linear orders on $G$ and $\delta$. We let $\As[P] = (A \times T,Q,\delta[P],I,F)$ where $\delta[P]$ is obtained by relabeling the transitions of \As as follows. Given $i \leq |\delta|$, if $(q_i,a_i,r_i) \in \delta$ is the $i$-th transition of \As, we replace it with the transition $(q_i,(a_i,t_i),r_i) \in \delta[P]$ where $t_i \in G$ is the $i$-th element of $G$ (recall that $|\delta| \leq |G|$ by hypothesis).

\begin{remark}
  A key property of $\As[P]$ is that, by definition, all transitions are labeled by distinct letters in $A \times T$. This implies that $K[\As,P] = L(\As[P])$ is recognized by a monoid of size at most $|Q|^2 + 2$.\end{remark}

We may now define the language $L[\As,P] \subseteq (A \cup \frE)^*$. Observe that we have a natural map $\mu: (A\frE^*)^* \to (A \times T)^*$. Indeed, consider $w \in (A\frE^*)^*$. Since $A \cap \frE = \emptyset$ (recall that this is a global assumption), it is immediate that $w$ admits a \emph{unique} decomposition $w = a_1w_1 \cdots a_n w_n$ with $a_1,\dots,a_n \in A$ and $w_1,\dots,w_n \in \frE^*$. Hence, we may define $\mu(w) = (a_1,P(w_1)) \cdots (a_n,P(w_n)) \in (A \times T)^*$. Finally, we define,
\[
  L[\As,P] = \frE^* \cdot  \mu\inv(K[\As,P]) \subseteq (A \cup \frE)^*
\]
We may now state the two key properties of $L[\As,P]$ upon which Theorem~\ref{thm:autoreduc} is based. It is recognized by a small monoid and the construction is connected to the separation.

\begin{proposition} \label{prop:variautored1}
  Given a \nfa $\As = (A,Q,\delta,I,F)$ and a compatible tagging $P$ of rank $n$, one may compute in \logspace a monoid morphism $\alpha: (A \cup \frE)^* \to M$ recognizing $L[\As,P]$ and such that $|M| \leq n + |A| \times n^2 \times (|Q|^2+2)$.
\end{proposition}

\begin{proposition} \label{prop:variautored2}
  Let $\Cs,\Ds$ be \pvaris such that \Ds extends \Cs. Consider two \nfas $\As_1$ and $\As_2$ over some alphabet $A$ and let $P$ be a compatible tagging that fools \Ds.  Then, $L(\As_1)$ is $\Cs(A)$-separable from $L(\As_2)$ if and only if $L[\As_1,P]$ is $\Ds(A \cup \frE)$-separable from $L[\As_2,P]$.
\end{proposition}

Let us explain why these two propositions imply Theorem~\ref{thm:autoreduc}. Let $\Cs,\Ds$ be \pvaris such that \Ds is smooth and extends \Cs. We show that the second assertion in the theorem holds (the first one is proved similarly).

Consider two \nfas $\As_i = (A,Q_j,\delta_j,I_j,F_j)$ for $j = 1,2$. We let $k= max(|\delta_1|,|\delta_2|)$. Since \Ds is smooth, we may compute (in \logspace) a tagging $P = (\tau: \frE^* \to T,G)$ of rank $|G| \geq k$. Then, we may use Proposition~\ref{prop:variautored1} to compute (in \logspace) monoid morphisms recognizing $L[\As_1,P]$ and $L[\As_2,P]$. Finally, by Proposition~\ref{prop:variautored2}, $L(\As_1)$ is $\Cs(A)$-separable from $L(\As_2)$ if and only if $L[\As_1,P]$ is $\Ds(A \cup \frE)$-separable from $L[\As_2,P]$. Altogether, this construction is a \logspace reduction to \Ds-separation for monoids which concludes the proof.

\subsection{Applications}

We now present the two main applications of Theorem~\ref{thm:autoreduc}. We start with the most simple one \pvaries. Indeed, we have the following lemma.

\begin{lemma} \label{lem:extendeasy}
  Let \Cs be a \pvarie. Then, \Cs is an extension of itself. Moreover, if $\bool{\Cs} \neq \reg$, then \Cs is smooth.
\end{lemma}

That a \pvarie is an extension of itself is immediate (one uses closure under inverse image). The difficulty is to prove smoothness. We may now combine Theorem~\ref{thm:autoreduc} with Lemma~\ref{lem:extendeasy} to get the following corollary.

\begin{corollary} \label{cor:autoreducvari}
  Let \Cs be a \pvarie such that $\bool{\Cs} \neq \reg$. There exists a \logspace reduction from \Cs-separation for \nfas to \Cs-separation for monoids.
\end{corollary}

Corollary~\ref{cor:autoreducvari} implies that for any \pvarie \Cs, the complexity of \Cs-separation is the same for monoids and \nfas. We illustrate this with an example: the \emph{star-free languages}.

\begin{example}
  Consider the star-free languages (\sfr): for every alphabet $A$, $\sfr(A)$ is the least set of languages containing all singletons $\{a\}$ for $a\in A$ and closed under Boolean operations and concatenation. It is folklore and simple to verify that \sfr is a \varie. It is known that \sfr-membership is in \nlog for monoids (this is immediate from Sch\"utzenberger's theorem~\cite{sfo}). On the other hand, \sfr-membership is \pspace-complete for \nfas. In fact, it is shown in~\cite{chofo} that \pspace-completeness still holds for \emph{deterministic} finite automata (\dfas).

  For \sfr-separation, we may combine Corollary~\ref{cor:autoreducvari} with existing results to obtain that the problem is in \exptime and \pspace-hard for both \nfas and monoids. Indeed, the \exptime upper bounds is proved in~\cite{pzfoj} for monoids and we may lift it to \nfas with Corollary~\ref{cor:autoreducvari}. Finally, the \pspace lower bound follows from~\cite{chofo}: \sfr-membership is \pspace-hard for \dfas. This yields that \sfr-separation is \pspace-hard for both \dfas and \nfas (by reduction from membership to separation which is easily achieved in \logspace when starting from a \dfa). Using Corollary~\ref{cor:autoreducvari} again, we get that \sfr-separation is \pspace-hard for monoids as well. \qed
\end{example}

We turn to our second application: finitely based concatenation hierarchies. Consider a finite \vari \Cs. We associate another finite \vari $\Cs_\frE$ which we only define for alphabets of the form $A \cup \frE$ (this is harmless: $\Cs_\frE$ is used as the output class of our reduction). Let $A$ be an alphabet and consider the morphism $\gamma: (A \cup \frE)^* \to A^*$ defined by $\gamma(a) = a$ for $a \in A$ and $\gamma(0) = \gamma(1) = \varepsilon$. We define,
\[
  \Cs_{\frE}(A \cup \frE) = \{\gamma\inv(L) \mid L \in \Cs(A)\}
\]
It is straightforward to verify that $\Cs_{\frE}$ remains a finite \vari. Moreover, we have the following lemma.

\begin{lemma} \label{lem:extendeasy2}
  Let \Cs be a finite \vari. For every $n \in \frac12 \nat$, $\Cs_{\frE}[n]$ is smooth and an extension of $\Cs[n]$.
\end{lemma}

In view of Theorem~\ref{thm:autoreduc}, we get the following corollary which provides a generic reduction for levels within finitely based hierarchies.

\begin{corollary} \label{cor:autoreduc}
  Let \Cs be a finite basis and $n \in \frac12 \nat$. There exists a \logspace reduction from $\Cs[n]$-separation for \nfas to $\Cs_{\frE}[n]$-separation for monoids.
\end{corollary}

\section{Generic upper bounds for low levels in finitely based hierarchies}
\label{sec:fixalph}
In this section, we present generic complexity results for the fixed alphabet separation problem associated to the lower levels in finitely based concatenation hierarchies. More precisely, we show that for every finite basis \Cs and every alphabet $A$, $\Cs[1/2](A)$- and $\Cs[1](A)$-separation are respectively in \nlog and in \ptime. These upper bounds hold for both monoids and \nfas: we prove them for monoids and lift the results to \nfas using the reduction of Corollary~\ref{cor:autoreduc}.

\begin{remark}
  We do \textbf{not} present new proofs for the decidability of $\Cs[1/2]$- and $\Cs[1]$-separation when \Cs is a finite \vari. These are difficult results which are proved in~\cite{pzbpol}. Instead, we recall the (inefficient) procedures which were originally presented in~\cite{pzbpol} and carefully analyze and optimize them in order to get the above upper bounds.
\end{remark}

For the sake of avoiding clutter, we fix an arbitrary finite \vari \Cs and an alphabet $A$ for the section.

\subsection{Key sub-procedure}

The algorithms $\Cs[1/2](A)$- and $\Cs[1](A)$-separation presented in~\cite{pzbpol} are based on a common sub-procedure. This remains true for the improved algorithms which we present in the paper. In fact, this sub-procedure is exactly what we improve to get the announced upper complexity bounds. We detail this point here.  Note that the algorithms require considering special monoid morphisms (called ``\Cs-compatible'') as input. We first define this notion.

\medskip
\noindent
{\bf \Cs-compatible morphisms.} Since \Cs is finite, one associates a classical equivalence $\sim_\Cs$ defined on $A^*$. Given $u,v \in A^*$, we write $u \sim_\Cs v$ if and only if $u \in L \ \Leftrightarrow\ v \in L$ for all $L \in \Cs(A)$. Given $w \in A^*$, we write $\ctype{w} \subseteq A^*$ for its $\sim_\Cs$-class. Since \Cs is a finite \vari, $\sim_\Cs$ is a congruence of finite index for concatenation (see~\cite{PZ:generic_csr_tocs:18} for a proof). Hence, the quotient ${A^*}/{\sim_\Cs}$ is a monoid and the map $w \mapsto \ctype{w}$ a morphism.

Consider a morphism $\alpha: A^* \to M$ into a finite monoid $M$. We say that $\alpha$ is \Cs-compatible when there exists a \emph{monoid morphism} $s \mapsto \ctype{s}$ from $M$ to ${A^*}/{\sim_\Cs}$ such that for every $w \in A^*$, we have $\ctype{w} = \ctype{\alpha(w)}$. Intuitively, the definition means that $\alpha$ ``computes'' the $\sim_\Cs$-classes of words in $A^*$. The following lemma is used to compute \Cs-compatible morphisms (note that the \logspace bound holds because \Cs and $A$ is fixed).

\begin{lemma} \label{lem:compat}
  Given two morphisms recognizing regular languages $L_1,L_2 \subseteq A^*$ as input, one may compute in \logspace a \Cs-compatible morphism which recognizes both $L_1$ and $L_2$.
\end{lemma}

In view of Lemma~\ref{lem:compat}, we shall assume in this section without loss of generality that our input in separation for monoids is a single \Cs-compatible morphism recognizing the two languages that need to be separated.

\medskip
\noindent
{\bf Sub-procedure.} Consider two \Cs-compatible morphisms $\alpha: A^* \to M$ and $\beta: A^* \to N$.  We say that a subset of $N$ is \emph{good} (for $\beta$) when it contains $\beta(A^*)$ and is closed under multiplication. For every good subset $S$ of $N$, we associate a subset of $M \times 2^N$. We then consider the problem of deciding whether specific elements belong to it (this is the sub-procedure used in the separation algorithms).

\begin{remark}
  The set $M \times 2^N$ is clearly a monoid for the componentwise multiplication. Hence we may multiply its elements and speak of idempotents in $M \times 2^N$.
\end{remark}

An \emph{$(\alpha,\beta,S)$-tree} is an unranked ordered tree. Each node $x$ must carry a label $lab(x) \in M \times 2^N$ and there are three possible kinds of nodes:
\begin{itemize}
\item {\bf Leaves}: $x$ has no children and $lab(x) = (\alpha(w),\{\beta(w)\})$ for some $w \in A^*$.
\item {\bf Binary}: $x$ has exactly two children $x_1$ and $x_2$. Moreover, if $(s_1,T_1) = lab(x_1)$ and $(s_2,T_2) = lab(x_2)$, then $lab(x) = (s_1s_2,T)$ with
  $T \subseteq T_1T_2$.
\item {\bf $S$-Operation}: $x$ has a unique child $y$. Moreover, the following must be satisfied:
  \begin{enumerate}
  \item The label $lab(y)$ is an idempotent $(e,E) \in M \times 2^N$.
  \item $lab(x) = (e,T)$ with $T \subseteq E \cdot \{t \in S \mid \ctype{e} = \ctype{t} \in S\} \cdot E$.
  \end{enumerate}
\end{itemize}
We are interested in deciding whether elements in $M \times 2^N$ are the root label of some computation tree. Observe that computing all such elements is easily achieved with a least fixpoint procedure: one starts from the set of leaf labels and saturates this set with three operations corresponding to the two kinds of inner nodes. This is the approach used in~\cite{pzbpol} (actually, the set of all root labels is directly defined as a least fixpoint and $(\alpha,\beta,S)$-trees are not considered). However, this is costly since the computed set may have exponential size with respect to $|N|$. Hence, this approach is not suitable for getting efficient algorithms. Fortunately, solving $\Cs[1/2](A)$- and $\Cs[1](A)$-separation does not require to have the whole set of possible root labels in hand. Instead, we shall only need to consider the elements $(s,T) \in M \times 2^N$ which are the root label of some tree \textbf{and} such that $T$ is a \textbf{singleton set}. It turns out that these specific elements can be computed efficiently. We state this in the next theorem which is the key technical result and main contribution of this section.

\begin{theorem} \label{thm:efficient}
  Consider two \Cs-compatible morphisms $\alpha: A^* \to M$ and $\beta: A^* \to N$ and a good subset $S \subseteq N$. Given $s \in M$ and $t \in N$, one may test in \nlog with respect to $|M|$ and $|N|$ whether there exists an $(\alpha,\beta,S)$-tree with root label $(s,\{t\})$.
\end{theorem}

Theorem~\ref{thm:efficient} is proved in appendix. We only present a brief outline which highlights two propositions about $(\alpha,\beta,S)$-trees upon which the theorem is based.

We first define a complexity measure for $(\alpha,\beta,S)$-trees. Consider two \Cs-compatible morphisms $\alpha: A^* \to M$ and $\beta: A^* \to N$ as well as a good subset $S \subseteq N$. Given an $(\alpha,\beta,S)$-tree \frT, we define the \emph{operational height of \frT} as the greatest number $h \in \nat$ such that \frT contains a branch with $h$ $S$-operation nodes.

Our first result is a weaker version of Theorem~\ref{thm:efficient}. It considers the special case when we restrict ourselves to $(\alpha,\beta,S)$-trees whose operational heights are bounded by a constant.

\begin{proposition} \label{prop:computimp}
  Let $h \in \nat$ be a constant and consider two \Cs-compatible morphisms $\alpha: A^* \to M$ and $\beta: A^* \to N$ and a good subset $S \subseteq N$. Given $s \in M$ and $t \in N$, one may test in \nlog with respect to $|M|$ and $|N|$ whether there exists an $(\alpha,\beta,S)$-tree  of operational height at most $h$ and with root label $(s,\{t\})$.
\end{proposition}

Our second result complements the first one: in Theorem~\ref{thm:efficient}, it suffices to consider $(\alpha,\beta,S)$-trees whose operational heights are bounded by a constant (depending only on the class \Cs and the alphabet $A$ which are fixed here). Let us first define this constant. Given a finite monoid $M$, we define the \Js-depth of $M$ as the greatest number $h \in \nat$ such that one may find $h$ pairwise distinct elements $s_1,\dots,s_h \in M$ such that for every $i < h$, $s_{i+1} = xs_iy$ for some $x,y \in M$

\begin{remark}
  The term ``\Js-depth'' comes from the Green's relations which are defined on any monoid~\cite{green}. We do not discuss this point here.
\end{remark}

Recall that the quotient set ${A^*}/{\sim_\Cs}$ is a monoid. Consequently, it has a \Js-depth. Our second result is as follows.

\begin{proposition} \label{prop:opbound}
  Let $h \in \nat$ be the \Js-depth of ${A^*}/{\sim_\Cs}$. Consider two \Cs-compatible morphisms $\alpha: A^* \to M$ and $\beta: A^* \to N$, and a good subset $S \subseteq N$. Then, for every $(s,T) \in M \times 2^N$, the following properties are equivalent:
  \begin{enumerate}
  \item $(s,T)$ is the root label of some $(\alpha,\beta,S)$-tree.
  \item $(s,T)$ is the root label of some $(\alpha,\beta,S)$-tree whose operational height is at most $h$.
  \end{enumerate}
\end{proposition}

In view of Proposition~\ref{prop:opbound}, Theorem~\ref{thm:efficient} is an immediate consequence of Proposition~\ref{prop:computimp} applied in the special case when $h$ is the $\Js$-depth of ${A^*}/{\sim_\Cs}$ and $m = 1$.

\subsection{Applications}

We now combine Theorem~\ref{thm:efficient} with the results of~\cite{pzbpol} to get the upper complexity bounds for $\Cs[1/2](A)$- and $\Cs[1](A)$-separation that we announced at the begging of the section.

\medskip
\noindent
{\bf Application to $\Cs[1/2]$.} Let us first recall the connection between $\Cs[1/2]$-separation and $(\alpha,\beta,S)$-trees. The result is taken from~\cite{pzbpol}.

\begin{theorem}[\cite{pzbpol}] \label{thm:poltheo}
  Let $\alpha: A^* \to M$ be a \Cs-compatible morphism and $F_0,F_1 \subseteq M$. Moreover, let $S = \alpha(A^*) \subseteq M$. The two following properties are equivalent:
  \begin{itemize}
  \item $\alpha\inv(F_0)$ is $\Cs[1/2]$-separable from $\alpha\inv(F_1)$.
  \item for every $s_0 \in F_0$ and $s_1 \in F_1$, there exists no $(\alpha,\alpha,S)$-tree with root label $(s_0,\{s_1\})$.
  \end{itemize}
\end{theorem}

By Theorem~\ref{thm:efficient} and the Immerman–Szelepcsényi theorem (which states that $\nlog = co\text{-}\nlog$), it is straightforward to verify that checking whether the second assertion in Theorem~\ref{thm:poltheo} holds can be done in \nlog with respect to $|M|$. Therefore, the theorem implies that $\Cs[1/2](A)$-separation for monoids is in \nlog. This is lifted to \nfas using Corollary~\ref{cor:autoreduc}.

\begin{corollary}  \label{cor:poltheo}
  For every finite basis \Cs and alphabet $A$, $\Cs[1/2](A)$-separation is in \nlog for both \nfas and monoids.
\end{corollary}

\medskip
\noindent
{\bf Application to $\Cs[1]$.} We start by recalling the $\Cs[1]$-separation algorithm which is again taken from~\cite{pzbpol}. In this case, we consider an auxiliary sub-procedure which relies on $(\alpha,\beta,S)$-trees.

Consider a \Cs-compatible morphism $\alpha: A^* \to M$. Observe that $M^2$ is a monoid for the componentwise multiplication. We let $\beta: A^* \to M^2$ as the morphism defined by $\beta(w) = (\alpha(w),\alpha(w))$ for every $w \in A^*$. Clearly, $\beta$ is \Cs-compatible: given $(s,t) \in M^2$, it suffices to define $\ctype{(s,t)} = \ctype{s}$. Using $(\alpha,\beta,S)$-trees, we define a procedure $S \mapsto Red(\alpha,S)$ which takes as input a good subset $S \subseteq M^2$ (for $\beta$) and outputs a subset $Red(\alpha,S) \subseteq S$.
\[
  Red(\alpha,S) = \{(s,t) \in S \mid \text{$(s,\{(t,s)\}) \in M \times 2^{M^2}$ is the root label of an $(\alpha,\beta,S)$-tree}\} \subseteq S
\]
It is straightforward to verify that $Red(\alpha,S)$ remains a good subset of $M^2$. We now have the following theorem which is taken from~\cite{pzbpol}.

\begin{theorem}[\cite{pzbpol}] \label{thm:bpoltheo}
  Let $\alpha: A^* \to M$ be a morphism into a finite monoid and $F_0,F_1 \subseteq M$. Moreover, let $S \subseteq M^2$ be the greatest subset of $\alpha(A^*) \times \alpha(A^*)$ such that $Red(\alpha,S) = S$. Then, the two following properties are equivalent:
  \begin{itemize}
  \item $\alpha\inv(F_0)$ is \bool{\pol{\Cs}}-separable from $\alpha\inv(F_1)$.
  \item for every $s_0 \in F_0$ and $s_1 \in F_1$, $(s_0,s_1) \not\in S$.
  \end{itemize}
\end{theorem}

Observe that Theorem~\ref{thm:efficient} implies that given an arbitrary good subset $S$ of $\alpha(A^*) \times \alpha(A^*)$, one may compute $Red(\alpha,S) \subseteq S$ in \ptime with respect to $|M|$. Therefore, the greatest subset $S$ of $\alpha(A^*) \times \alpha(A^*)$ such that $Red(\alpha,S) = S$ can be computed in \ptime using a greatest fixpoint algorithm. Consequently, Theorem~\ref{thm:bpoltheo} yields that $\Cs[1](A)$-separation for monoids is in \ptime. Again, this is lifted to \nfas using Corollary~\ref{cor:autoreduc}.

\begin{corollary}  \label{cor:bpoltheo}
  For every finite basis \Cs and alphabet $A$, $\Cs[1](A)$-separation is in \ptime for both \nfas and monoids.
\end{corollary}

\section{The Straubing-Thérien hierarchy}
\label{sec:classic}
In this final section, we consider one of the most famous concatenation hierarchies: the Straubing-Thérien hierarchy~\cite{StrauConcat,TheConcat}. We investigate the complexity of separation for the levels 3/2 and 2.

\begin{remark}
  Here, the alphabet is part of the input. For fixed alphabets, these levels can be handled with the generic results presented in the previous section (see Theorem~\ref{thm:alphatrick} below).
\end{remark}

The basis of the Straubing-Thérien hierarchy is the trivial \varie \stzer defined by $\stzer(A) = \{\emptyset,A^*\}$ for every alphabet $A$. It is known and simple to verify (using induction) that all half levels are \pvaries and all full levels are \varies.

The complexity of separation for the level one (\stone) has already been given a lot of attention. Indeed, this level corresponds to a famous class which was introduced independently from concatenation hierarchies: the piecewise testable languages~\cite{simon75}. It was shown independently in~\cite{martens} and~\cite{pvzmfcs13} that \stone-separation is in \ptime for \nfas (and therefore for \dfas and monoids as well). Moreover, it was also shown in~\cite{Masopust18} that the problem is actually \ptime-complete for \nfas and \dfas\footnote{Since \stone is a \varie, \ptime-completeness for \stone-separation can also be lifted to monoids using Corollary~\ref{cor:autoreducvari}.}. Additionally, it is shown in~\cite{martens} that \sthone-separation is in \nlog.

In the paper, we are mainly interested in the levels \sthtwo and \sttwo. Indeed, the Straubing-Thérien hierarchy has a unique property: the generic separation results of~\cite{pzbpol} apply to these two levels as well. Indeed, these are also the levels 1/2 and 1 in another finitely based hierarchy. Consider the class \at of \emph{alphabet testable languages}. For every alphabet $A$, $\at(A)$ is the set of all Boolean combinations of languages $A^*a A^*$ for $a \in A$. One may verify that \at is a \varie and that $\at(A)$ is finite for every alphabet $A$. Moreover, we have the following theorem which is due to Pin and Straubing~\cite{pin-straubing:upper} (see~\cite{PZ:generic_csr_tocs:18} for a modern proof).

\begin{theorem}[\cite{pin-straubing:upper}] \label{thm:alphatrick}
  For every $n \in \frac12 \nat$, we have $\at[n] = \sttp{n+1}$.
\end{theorem}

The theorem implies that $\sthtwo = \at[1/2]$ and $\sttwo = \at[1]$. Therefore, the results of~\cite{pzbpol} yield the decidability of separation for both \sthtwo and \sttwo (the latter is the main result of~\cite{pzbpol}). As expected, this section investigates complexity for these two problems.

\subsection{The level 3/2}

We have the following tight complexity bound for \sthtwo-separation.

\begin{theorem} \label{thm:sth}
  \sthtwo-separation is \pspace-complete for both \nfas and monoids.
\end{theorem}

The \pspace upper bound is proved by building on the techniques introduced in the previous section for handling the level 1/2 of an arbitrary finitely based hierarchies. Indeed, we have $\sthtwo = \at[1/2]$ by Theorem~\ref{thm:alphatrick}. However, let us point out that obtaining this upper bound requires some additional work: the results of Section~\ref{sec:fixalph} apply to the setting in which the alphabet is fixed, this is not the case here. In particular, this is why we end up with a \pspace upper bound instead of the generic \nlog upper presented in Corollary~\ref{cor:poltheo}. The detailed proof is postponed to the appendix.

\medskip

In this abstract, we focus on proving that \sthtwo-separation is \pspace-hard. The proof is presented for \nfas: the result can then be lifted to monoids with Corollary~\ref{cor:autoreducvari} since \sthtwo is a \pvarie. We use a \logspace reduction from the quantified Boolean formula problem (QBF) which is among the most famous \pspace-complete problems.

We first describe the reduction. For every quantified Boolean formula $\Psi$, we explain how to construct two languages $L_\Psi$ and $L'_\Psi$. It will be immediate from the presentation that given $\Psi$ as input, one may compute \nfas for $L_\Psi$ and $L'_\Psi$ in \logspace. Then, we show that this construction is the desired reduction: $\Psi$ is true if and only if $L_\Psi$ is not \sthtwo-separable from $L'_\Psi$.

\medskip

Consider a quantified Boolean formula $\Psi$ and let $n$ be the number of variables it involves. We assume without loss of generality that $\Psi$ is in prenex normal form and that the quantifier-free part of $\Psi$ is in conjunctive normal form (QBF remains \pspace-complete when restricted to such formulas). That is,
\[
  \Psi = Q_n\ x_n \cdots Q_1\ x_1\ \varphi
\]
where $x_1 \dots x_n$ are the variables of $\Psi$, $Q_1,\dots,Q_n \in \{\exists,\forall\}$ are quantifiers and $\varphi$ is a quantifier-free Boolean formula involving the variables $x_1 \dots x_n$ which is in conjunctive normal form.

We describe the two regular languages $L_\Psi,L'_\Psi$ by providing regular expressions recognizing them. Let us first specify the alphabet over which these languages are defined. For each variable $x_i$  occurring in $\Psi$, we create two letters that we write $x_i$ and $\overline{x_i}$. Moreover, we let,
\[
  X = \{x_1,\dots,x_n\} \quad \text{and} \quad \overline{X} = \{\overline{x_1},\dots,\overline{x_n}\}
\]
Additionally, our alphabet also contains the following letters: $\#_1,\dots,\#_i,\$$. For $0 \leq i \leq n$, we define an alphabet $B_i$. We have:
\[
  B_0 = X \cup \overline{X} \quad \text{and} \quad B_i = X \cup \overline{X} \cup \{\#_1,\dots,\#_i,\$\}
\]
Our languages are defined over the alphabet $B_n$: $L_\Psi,L_\Psi' \subseteq B_n^*$. They are built by induction: for $0 \leq i \leq n$ we describe two languages $L_i,L'_i \subseteq B_i^*$ (starting with the case $i = 0$). The languages $L_\Psi,L_\Psi'$ are then defined as $L_n,L'_n$.

\medskip
\noindent
{\bf Construction of $L_0,L'_0$.} The language $L_0$ is defined as $L_0 = (B_0)^*$. The language $L'_0$ is defined from the quantifier-free Boolean formula $\varphi$.
Recall that by hypothesis $\varphi$ is in conjunctive normal form: $\varphi = \bigwedge_{j \leq k} \varphi_j$ were $\varphi_i$ is a disjunction of literals. For all $j \leq k$, we let $C_j \subseteq B_0 = X  \cup \overline{X}$ as the following alphabet:
\begin{itemize}
\item Given $x \in X$, we have $x \in C_j$, if and only $x$ is a literal in the disjunction $\varphi_j$.
\item Given $\overline{x} \in \overline{X}$, we have $\overline{x} \in C_j$, if and only $\neg x$ is a literal in the disjunction $\varphi_j$.
\end{itemize}
Finally, we define $L'_0 = C_1C_2 \cdots C_k$.

\medskip
\noindent
{\bf Construction of $L_i,L'_i$ for $i \geq 1$.} We assume that $L_{i-1},L'_{i-1}$ are defined and describe $L_i$ and $L'_i$. We shall use the two following languages in the construction:
\[
  T_i = (\#_ix_i (B_{i-1} \setminus \{\overline{x_i}\})^*\$x_i)^* \quad \text{and} \quad  \overline{T_i} = (\#_i \overline{x_i} (B_{i-1} \setminus \{x_i\})^*\$\overline{x_i})^*
\]
The definition of $L_i,L'_i$ from $L_{i-1},L'_{i-1}$ now depends on whether the quantifier $Q_i$ is existential or universal. \begin{itemize}
\item If $Q_i$ is an existential quantifier (i.e. $Q_i = \exists$):
  \[
    \begin{array}{lll}
      L_i & = & (\#_i(x_i + \overline{x_i})L_{i-1}\$(x_i + \overline{x_i}))^*\#_i \\
      L'_i & = &(\#_i (x_i + \overline{x_i})L'_{i-1}\$(x_i + \overline{x_i}))^* \#_i\$ \left(T_i\#_i + \overline{T_i}\#_i\right)
    \end{array}
  \]
\item If the $Q_i$ is an universal quantifier (i.e. $Q_i = \forall$):
  \[
    \begin{array}{lll}
      L_i & = & (\#_i(x_i + \overline{x_i})L_{i-1}\$(x_i + \overline{x_i}))^*\#_i \\
      L'_i & = & \overline{T_i}\#_i\$(\#_i (x_i + \overline{x_i})L'_{i-1}\$(x_i + \overline{x_i}))^* \#_i\$ T_i\#_i
    \end{array}
  \]
\end{itemize}

Finally, $L_\Psi,L_\Psi'$ are defined as the languages $L_n,L'_n \subseteq (B_n)^*$. It is straightforward to verify from the definition, than given $\Psi$ as input, one may compute \nfas for $L_\Psi$ and $L_\Psi'$ in \logspace. Consequently, it remains to prove that this construction is the desired reduction. We do so in the following proposition.

\begin{proposition} \label{prop:reducgoal}
  For every quantified Boolean formula $\Psi$, $\Psi$ is true if and only if $L_\Psi$ is not \sthtwo-separable from $L'_\Psi$.
\end{proposition}

Proposition~\ref{prop:reducgoal} is proved by considering a stronger result which states properties of all the languages $L_i,L'_i$ used in the construction of  $L_\Psi,L_\Psi'$ (the argument is an induction on $i$). While we postpone the detailed proof to the appendix, let us provide a sketch which presents this stronger result.

\begin{proof}[Proof of Proposition~\ref{prop:reducgoal} (sketch)]
  Consider a quantified Boolean formula $\Psi$. Moreover, let $B_0,\dots,B_n$ and $L_i,L'_i \subseteq (B_i)^*$ as the alphabets and languages defined above. The key idea is to prove a property which makes sense for all languages $L_i,L'_i$. In the special case when $i = n$, this property implies Proposition~\ref{prop:reducgoal}.

  Consider $0 \leq i \leq n$. We write $\Psi_i$ for the sub-formula $\Psi_i := Q_i\ x_i \cdots Q_1\ x_1\ \varphi$ (with the free variables $x_{i+1},\dots,x_n$). In particular, $\Psi_0 := \varphi$ and $\Psi_n := \Psi$. Moreover, we call ``\emph{$i$-valuation}'' a sub-alphabet $V \subseteq B_i$ such that,
  \begin{enumerate}
  \item $\#_1,\dots,\#_i,\$ \in V$ and $x_1,\overline{x_1},\dots,x_i,\overline{x_i} \in V$, and,
  \item for every $j$ such that $i < j \leq n$, one of the two following property holds:
    \begin{itemize}
    \item $x_j \in V$ and $\overline{x_j} \not\in V$, or,
    \item $x_j \not\in V$ and $\overline{x_j} \in V$.
    \end{itemize}
  \end{enumerate}
  Clearly, an $i$-valuation corresponds to a truth assignment for all variables $x_j$ such that $j > i$ (i.e. those that are free in $\Psi_i$): when the first (resp. second) assertion in Item~2 holds, $x_j$ is assigned to $\top$ (resp. $\bot$).  Hence, abusing terminology, we shall say that an $i$-valuation $V$ \emph{satisfies} $\Psi_i$ if $\Psi_i$ is true when replacing its free variables by the truth values provided by $V$.

  Finally, for $0 \leq i \leq n$, if $V \subseteq B_i$ is an $i$-valuation, we let $[V] \subseteq V^*$ as the following language. Given $w \in V^*$, we have $w \in [V]$ if and only if for every $j > i$ either $x_j \in \cont{w}$ or $\overline{x_j} \in \cont{w}$ (by definition of $i$-valuations, exactly one of these two properties must hold). Proposition~\ref{prop:reducgoal} is now a consequence of the following lemma.

  \begin{lemma} \label{lem:reduclem}
    Consider $0 \leq i \leq n$. Then given an $i$-valuation $V$, the two following properties are equivalent:
    \begin{enumerate}
    \item $\Psi_i$ is satisfied by $V$.
    \item $L_i \cap [V]$ is not \sthtwo-separable from $L'_i \cap [V]$.
    \end{enumerate}
  \end{lemma}

  Lemma~\ref{lem:reduclem} is proved by induction on $i$ using standard properties of the polynomial closure operation (see~\cite{PZ:generic_csr_tocs:18} for example). The proof is postponed to the appendix. Let us explain why the lemma implies Proposition~\ref{prop:reducgoal}.

  Consider the special case of Lemma~\ref{lem:reduclem} when $i = n$. Observe that $V = B_n$ is an $n$-valuation (the second assertion in the definition of $n$-valuations is trivially true since there are no $j$ such that $n < j \leq n$). Hence, since $\Psi = \Psi_n$ and $L_\Psi,L'_\Psi = L_n,L'_n$, the lemma yields that,
  \begin{enumerate}
  \item $\Psi$ is satisfied by $V$ (i.e. $\Psi$ is true).
  \item $L_\Psi \cap [V]$ is not \sthtwo-separable from $L'_\Psi \cap [V]$.
  \end{enumerate}
  Moreover, we have $[V] = (B_n)^*$ by definition. Hence, we obtain that $\Psi$ is true if and only if $L$ is not \sthtwo-separable from $L'$ which concludes the proof of Proposition~\ref{prop:reducgoal}.
\end{proof}

\subsection{The level two}

For the level two, there is a gap between the lower and upper bound that we are able to prove. Specifically, we have the following theorem.

\begin{theorem} \label{thm:st}
  \sttwo-separation is in \exptime and \pspace-hard for both \nfas and monoids.
\end{theorem}

Similarly to what happened with \sthtwo, the \exptime upper bound is obtained by building on the techniques used in the previous section. Proving \pspace-hardness is achieved using a reduction from \sthtwo-separation (which is \pspace-hard by Theorem~\ref{thm:sth}). The reduction is much simpler than what we presented for \sthtwo above. It is summarized by the following proposition.

\begin{proposition} \label{prop:bpolred}
  Consider an alphabet $A$ and $H,H' \subseteq A^*$. Let $B = A \cup \{\#,\$\}$ with $\#,\$ \not\in A$, $L = \#(H'\#(A^*\$\#)^*)^*H\#(A^*\$\#)^* \subseteq B^*$ and $L' = \#(H'\#(A^*\$\#)^*)^* \subseteq B^*$. The two following properties are equivalent:
  \begin{enumerate}
  \item $H$ is \sthtwo-separable from $H'$.
  \item $L$ is \sttwo-separable from $L'$.
  \end{enumerate}
\end{proposition}

Proposition~\ref{prop:bpolred} is proved using standard properties of the polynomial and Boolean closure operations. The argument is postponed ot the appendix. It is clear than given as input \nfas for two languages $H,H'$, one may compute \nfas for the languages $L,L'$ defined Proposition~\ref{prop:bpolred} in \logspace. Consequently, the proposition yields the desired \logspace reduction from \sthtwo-separation for \nfas to \sttwo-separation for \nfas. This proves that \sttwo-separation is \pspace-hard for \nfas (the result can then be lifted to monoids using Corollary~\ref{cor:autoreducvari}) since \sttwo is a \varie).

\section{Conclusion}
\label{sec:conc}
We showed several results, all of them raising new questions. First we proved that for many important classes of languages (including all \pvaries), the complexity of separation does not depend on how the input languages are represented. A natural question is whether the technique can be adapted to encompass more classes. In particular, one may define more permissive notions of \pvaries by replacing  closure under inverse image by weaker notions. For example, many natural classes are \emph{length increasing \pvaries}: closure under inverse image only has to hold for length increasing morphisms (\emph{i.e.}, morphisms $\alpha: A^* \to B^*$ such that $|\alpha(w)| \geq |w|$ for every $w \in A^*$). For example, the levels of another famous concatenation hiearchy, the dot-depth~\cite{BrzoDot} (whose basis is $\{\emptyset,\{\varepsilon\},A^+,A^*\}$) are length increasing \pvaries. Can our techniques be adapted for such classes? Let us point out that there exists no example of natural class \Cs for which separation is decidable and strictly harder for \nfas than for monoids. However, there are classes \Cs for which the question is open (see for example the class of locally testable languages in~\cite{pvzltt}).

We also investigated the complexity of separation for levels 1/2 and 1 in finitely based concatenation hierarchies. We showed that when the alphabet is fixed, the problems are respectively in \nlog and \ptime for any such hierarchy. An interesting follow-up question would be to push these results to level 3/2, for which separation is also known to be decidable in any finitely based concatenation hierarchy~\cite{pbp}. A rough analysis of the techniques used in~\cite{pbp} suggests that this requires moving above \ptime.

Finally, we showed that in the famous Straubing-Thérien hierarchy, \sthtwo-separation is \pspace-complete and \sttwo-separation is in \exptime and \pspace-hard. Again, a natural question is to analyze \sththree-separation whose decidability is established in~\cite{pbp}.

% \bibliography{main}

\appendix
\newpage

\section{Appendix to Section~\ref{sec:nfatomono}}
\label{app:nfatomono}
In this appendix, we present the missing proofs for the statements of Section~\ref{sec:nfatomono}.

\subsection{Proof of Proposition~\ref{prop:variautored1}}

We start with Proposition~\ref{prop:variautored1} which is used to build morphisms recognizing the languages we associate to \nfas and tagging pairs. Let us recall the statement.

\adjustc{prop:variautored1}

\begin{proposition}

  Given a \nfa $\As = (A,Q,\delta,I,F)$ and a compatible tagging $P$ of size $n$, one may compute in \logspace a monoid morphism $\alpha: (A \cup \frE)^* \to M$ recognizing $L[\As,P]$ and such that $|M| \leq n + |A| \times n^2 \times (|Q|^2+2)$.

\end{proposition}

\restorec

Let $P = (\tau: \frE^* \to T,G)$ ($n = |T|$). We construct the morphism $\alpha: (A \cup \frE)^* \to M$ recognizing $L[\As,P] \subseteq (A \cup \frE)^*$. That it has size $|M| \leq n + |A| \times n^2 \times (|Q|^2+2)$ and can be computed in \logspace is immediate from the construction.

Recall that $L[\As,P]$ is defined from an intermediary language $K[\As,P] \subseteq (A \times T)^*$ which is recognized by the \nfa $\As[P]$. We first prove the following preliminary result about $K[\As,P]$ which uses the fact that, by construction, all transitions in $\As[P]$ are labeled by distinct letters in $A \times T$.

\begin{lemma} \label{lem:variautored1}

  The language $K[\As,P]$ is recognized by a morphism $\beta: (A \times T)^* \to N$ such that monoid $N$ has size $|N| \leq |Q|^2+2$.

\end{lemma}

\begin{proof}

  Recall that $\As[P] = (A \times T,Q,\delta[P],I,F)$ where $\delta[P]$ is obtained by relabeling the transition of \As. We let $N = Q^2 \cup \{0_N,1_N\}$ and equip $N$ with the following multiplication. The elements $0_N$ and $1_N$ are respectively a zero and a neutral element. For $(q_1,r_1),(q_2,r_2) \in Q^2$, we define,

  \[
    (q_1,r_1) \cdot (q_2,r_2) = \left\{
      \begin{array}{ll}
        (q_1,r_2) & \text{if $r_1 = q_2$} \\
        0_N & \text{otherwise}
      \end{array}
    \right.
  \]

  We now define a morphism $\beta: (A \times T)^* \to N$. Given $(a,t) \in A \times T$, we know by definition that there exists at most one transition in $\delta[P]$ whose label is $(a,t)$. Therefore, either there is no such transition and we let $\beta((a,t)) = 0_N$ or there exists exactly one pair $(q,r) \in Q^2$ such that $(q,(a,t),r) \in \delta[P]$ and we define $\beta((a,t)) = (q,r)$. One may now verify that $\beta$ recognizes $L(\As[P]) = K[\As,P]$.

\end{proof}

Let us briefly recall how $L[\As,P] \subseteq (A \cup \frE)^*$ is defined from $K[\As,P]$. We have a map $\mu: (A\frE^*)^* \to (A \times T)^*$ defined as follows. Consider $w \in (A\frE^*)^*$. Since $A \cap \frE = \emptyset$, $w$ admits a unique decomposition $w = a_1w_1 \cdots a_n w_n$ with $a_1,\dots,a_n \in A$ and $w_1,\dots,w_n \in \frE^*$. We define, $\mu(w) = (a_1,\tau(w_1)) \cdots (a_n,\tau(w_n))$. Finally, recall that,
\[
  L[\As,P] = \frE^* \cdot \mu\inv(K[\As,P]) \subseteq \frE^*(A\frE^*)^* = (A \cup \frE)^*
\]
We may now define the morphism $\alpha: (A \cup \frE)^* \to M$. We let $\beta: (A \times T)^* \to N$ as the morphism given by Lemma~\ref{lem:variautored1}. Consider the following set $M$:
\[
  M = T \cup (T \times N \times A \times T)
\]
Note that since $|N| \leq |Q|^2+2$, we do have $|M| \leq n + |A| \times n^2 \times (|Q|^2+2)$ as desired. We equip $M$ with the following multiplication. Since $M$ is defined as a union there are two kinds of elements which means that we have to consider four cases:

\begin{itemize}

\item If $t,t' \in T$, then their multiplication as element of $M$ is the one in $T$, i.e. $tt'$.

\item If $t  \in T$ and $(t_1,s,a,t_2) \in T \times N \times A \times T$, we let,

  \[
    \begin{array}{lll}
      t \cdot (t_1,s,a,t_2) & = & (tt_1,s,a,t_2) \\
      (r,t_1,s,a,t_2) \cdot t & = & (t_1,s,a,t_2t)
    \end{array}
  \]

\item If $(t_1,s,a,t_2),(t'_1,s',a',t'_2) \in T \times N \times A \times T$, we let,

  \[
    (t_1,s,a,t_2) \cdot (t'_1,s',a',t'_2) = (t_1,s \beta((a,t_2t'_1))s',a',t'_2)
  \]
\end{itemize}

One may verify that this multiplication is associative and that $1_T \in T$ is a neutral element for $M$. Finally, we define a morphism $\alpha: (A \cup \frE)^* \to M$ as follows. For $a \in A$, we let $\alpha(a) = (1_T,1_N,a,1_T) \in T \times N \times A \times T$ and for $b \in \frE$, we let $\alpha(b) = \tau(b) \in T$. The following fact can be verified from the definition of $\alpha$.

\begin{fact} \label{fct:thealpha}
  Consider a word $u \in (A \cup \frE)^*$. Then, one of the two following properties holds:

  \begin{enumerate}
  \item $u \in \frE^*$ and $\alpha(u) = \tau(u) \in T$.
  \item $u = u_0u_1au_2$ with $u_0 \in \frE^*$, $u_1 \in (A\frE^*)^*$, $a \in A$ and $u_2 \in \frE^*$ and we have,
    \[
      \alpha(u) = (\tau(u_0),\beta(\mu(u_1)),a,\tau(u_2)).
    \]
  \end{enumerate}
\end{fact}

It remains to verify that $\alpha$ recognizes $L[\As,P]$. Since $K[\As,P]$ is recognized by $\beta$, we have $H \subseteq N$ such that $K[\As,P] = \beta\inv(H)$. We define $H' \subseteq M$ as the following set:
\[
  H' = \left\{
    \begin{array}{ll}

      \{(t_1,s,a,t_2) \in T \times N \times A \times T \mid s \beta((a,t_2)) \in H\} & \text{if $1_N \not\in H$} \\

      \{(t_1,s,a,t_2) \in T \times N \times A \times T \mid s \beta((a,t_2)) \in H\} \cup T & \text{if $1_N \in H$}

    \end{array}
  \right.
\]
Since $L[\As,P] = \frE^* \cdot \mu\inv(K[\As,P])$ by definition, it can be verified from Fact~\ref{fct:thealpha} that $L[\As,P] = \alpha\inv(H')$ which concludes the proof.

\subsection{Proof of Proposition~\ref{prop:variautored2}}

We first recall Proposition~\ref{prop:variautored2}.

\adjustc{prop:variautored2}

\begin{proposition}

  Let $\Cs,\Ds$ be \pvaris such that \Ds extends \Cs. Consider two \nfas $\As_1$ and $\As_2$ over some alphabet $A$ and let $P$ be a compatible tagging that fools \Ds.  Then, $L(\As_1)$ is $\Cs(A)$-separable from $L(\As_2)$ if and only if $L[\As_1,P]$ is $\Ds(A \cup \frE)$-separable from $L[\As_2,P]$.

\end{proposition}

\restorec

We fix $\As_1 = (A,Q_1,\delta_1,I_1,F_1)$ and $\As_2 = (A,Q_2,\delta_2,I_2,F_2)$ for the proof. Moreover, we let $P = (\tau: \frE^* \to T,G)$ as the tagging pair which fools \Ds.

\medskip

There are two directions to prove. First, we assume that $L(\As_1)$ is \Cs-separable from $L(\As_2)$. We prove that $L[\As_1,P]$ is \Ds-separable from $L[\As_2,P]$. Note that this direction is independent from the hypothesis that $P$ fools \Ds. Let $K \in \Cs(A)$ be a separator for $L(\As_1)$ and $L(\As_2)$: $L(\As_1) \subseteq K$ and $L(\As_2) \cap K = \emptyset$. Consider the morphism $\gamma: (A \cup \frE)^* \to A^*$ defined by $\gamma(a) = a$ for $a \in A$ and $\gamma(b) = \varepsilon$ for $b \in \frE$. Since \Ds is an extension of \Cs, we have $\gamma\inv(K) \in \Ds(A \cup \frE)$ by definition. Moreover, it is straightforward to verify from the definitions of $\gamma$, $L[\As_1,P]$ and $L[\As_2,P]$ that $\gamma\inv(K)$ separates $L[\As_1,P]$ from $L[\As_2,P]$ which concludes this direction of the proof.

\medskip

Assume now that $L[\As_1,P]$ is \Ds-separable from $L[\As_2,P]$. We show that $L(\As_1)$ is \Cs-separable from $L(\As_2)$. Let $K \in \Ds(A \cup \frE)$ which separates $L[\As_1,P]$ from $L[\As_2,P]$. Clearly, $K \in \bool{\Ds}(A \cup \frE)$. Moreover, since \Ds is a \pvari, one may verify that $\bool{\Ds}$ is a \vari (quotients commute with Boolean operations). Therefore, it follows from standard results about \varis that there exists a morphism $\alpha: (A \cup \frE)^* \to M$ into a finite monoid $M$ which recognizes $K$ and such that every language recognized by $\alpha$ belongs to \bool{\Ds} (it suffices to choose $\alpha$ as the ``syntactic morphism'' of $K$, see~\cite{pingoodref} for details). By definition of $\alpha$ and since $P$ fools \Ds, we get the following fact.

\begin{fact} \label{fct:themorphism}
  There exists $s \in M$ such that for every $t \in G$, we have $w_t \in \frE^*$ satisfying $\alpha(w_t) = s$ and $\tau(w_t) = t$.
\end{fact}

Let $u = w_t \in \frE^*$ for some arbitrary $t \in G$ and consider the morphism $\lambda_u: A^* \to (A \cup \frE)^*$ defined by $\gamma(a) = au \in (A \cup \frE)^*$ for every $a \in A$. Finally, we let $K' = \lambda_u\inv(K)$. Since $K \in \Ds(A \cup \frE)$ and \Ds is an extension of \Cs, it is immediate that $K' \in \Cs(A)$. We now show that $K'$ separates $L(\As_1)$ from $L(\As_2)$ which concludes the argument.

We concentrate on proving that $L(\As_1) \subseteq K'$. That $L(\As_2) \cap K' = \emptyset$ is showed symmetrically and left to the reader. Consider some word $v = a_1 \cdots a_n \in L(\As_1)$. We show that $v \in K'$. By definition of $L[\As_1,P]$, it is straightforward to verify that there exists $t_1,\dots,t_n \in G$ (each depending on the whole word $v$) such that $a_1w_{t_1} \cdots a_nw_{t_n} \in L[\As_1,P]$. Moreover, by definition in Fact~\ref{fct:themorphism}, we know that $\alpha(w_t) = \alpha(u) = s$ for every $t \in G$. Consequently, we get,
\[
  \alpha(a_1w_{t_1} \cdots a_nw_{t_n}) = \alpha(a_1u \cdots a_nu) = \alpha(\lambda_u(v))
\]

Since $\alpha$ recognizes $L[\As_1,P]$ which contains $a_1w_{t_1} \cdots a_nw_{t_n}$, it follows that $\lambda_u(v) \in L[\As_1,P]$ as well. Hence, since $L[\As_1,P] \subseteq K$, we obtain that $\lambda_u(v) \in K$. Finally, this yields $v \in \lambda_u\inv(K) = K'$, finishing the proof.

\subsection{Proof of Lemma~\ref{lem:extendeasy}}

We first recall the statement of Lemma~\ref{lem:extendeasy}.

\adjustc{lem:extendeasy}

\begin{lemma}

  Let \Cs be a \pvarie. Then, \Cs is an extension of itself. Moreover, if $\bool{\Cs} \neq \reg$, then \Cs is smooth.

\end{lemma}

\restorec

We fix the \pvarie \Cs for the proof. Clearly, \Cs is an extension of itself since \pvaries are closed under inverse image by definition. We now assume that $\bool{\Cs} \neq \reg$ and show that \Cs is smooth: given as input $k \in \nat$, one may compute in \logspace (with respect to $k$) a tagging of rank at least $k$ and which fools \Cs. We describe how to construct a tagging of rank $k$ and size polynomial in $k$, that it can be computed in \logspace is straightforward to verify and left to the reader. Furthermore, we consider the special case when $k = 2^h$ for some $h \geq 1$ (when $k$ is not of this form, it suffices to consider the least $h$ such that $k \leq 2^h$). The construction is based on the following preliminary lemma.

\begin{lemma} \label{lem:tags}

  There exist constants $\ell,m \in \nat$ such that for every $h \geq 1$, there exists a morphism $\gamma: B^* \to T$ and $F \subseteq T$ such that,

  \begin{enumerate}

  \item $B \leq h \times \ell$, $|T| \leq m^h$ and $|F| \geq 2^h$.

  \item for every alphabet $A$ and every morphism $\alpha: (A \cup B)^* \to M$ into a finite monoid $M$, if all languages recognized by $\alpha$ belongs to $\bool{\Cs}(A \cup B)$, then, there exists $s \in M$, such that for every $t \in T$, we have $w_t \in B^*$ which satisfies $\alpha(w_t) = s$ and $\tau(w_t) = t$.

  \end{enumerate}

\end{lemma}

Before we prove Lemma~\ref{lem:tags}, let us use it to finish the construction of smooth taggings. We fix $h \geq 1$ and build a tagging of rank $2^h$ and size polynomial in $2^h$. Let $\gamma: B^* \to T$ and $F \subseteq T$ be as defined in Lemma~\ref{lem:tags}. We fix some binary encoding of the alphabet $B$ over the two letter alphabet \frE given by the morphism $\eta: B^* \to \frE*$: for every $b \in B$, $\eta(b)$ is distinct word of length $log_2(|B|)$.

It is straightforward to build a morphism $\tau: \frE^* \to T'$ which recognizes the languages $\eta(\gamma\inv(s))$ for $s \in T$.   Moreover, one may verify that it is possible to do so with a monoid $T'$ of size polynomial with respect to $|T|$ and $|B|$. Therefore the size of $T'$ is polynomial with respect to $2^h$ since $B \leq h \times m$, $|T| \leq m^h$. One may now verify from our hypothesis on $\gamma$ that there exists $F' \subseteq T'$ such that $|F'| \geq 2^h$ and $(\tau: \frE^* \to T',F')$ fools \Cs. This concludes the main proof. It remains to handle Lemma~\ref{lem:tags}.

\begin{proof}[Proof of Lemma~\ref{lem:tags}]

  We start by proving the following fact which handles the special case when $h = 1$. We shall use this fact to define the constants $\ell,m \in \nat$.

  \begin{fact} \label{fct:tags}

    There exists a morphism $\eta: D^* \to R$ and $G \subseteq R$ such that $|G| = 2$ and for every alphabet $A$ and every morphism $\alpha: (A \cup D)^* \to M$ into a finite monoid $M$, if all languages recognized by $\alpha$ belongs to $\bool{\Cs}(A \cup D)$, then, there exists $s \in M$, such that for every $r \in R$, we have $w_r \in D^*$ which satisfies $\alpha(w_r) = s$ and $\eta(w_t) = t$.

  \end{fact}

  \begin{proof}

    Since $\bool{\Cs} \neq \reg$, there exist an alphabet $D$ and a regular language $L \subseteq D^*$ such that $L \not\in \bool{\Cs}(D)$. Since $L$ is regular, we have a morphism $\eta: D^* \to R$ into a finite monoid $R$  and $XF \subseteq R$ such that $L = \eta\inv(X)$. Since $L \not\in \bool{\Cs}$, it is not \bool{\Cs}-separable from $D^* \setminus L = \eta\inv(R \setminus X)$. This implies the existence of $r \in X$ and $r' \in R \setminus X$ such that $\eta\inv(r)$ is not \bool{\Cs}-separable from $\eta\inv(r')$. We let $G = \{r,r'\}$. It remains to show the property described in the fact is satisfied.

    Consider a morphism $\alpha: (A \cup D)^* \to M$ such that every language recognized by $\alpha$ belongs to $\bool{\Cs}(A \Cup D)$. We have to exhibit $s \in M$ and $w,w' \in D^*$ such that $\alpha(w) = \alpha(w') = s$, $\eta(w) = r$ and $\eta(w') = r'$. Let $\beta: D^* \to M$ be the restriction of $\alpha$ to $D^*$. Since \bool{\Cs} is a \varie, one may verify that every language recognized by $\beta$ belongs to $\bool{\Cs}(D)$. Since $\eta\inv(r) \subseteq D^*$ is not \bool{\Cs}-separable from $\eta\inv(r') \subseteq D^*$, it follows that there exists $s \in M$ such that $\beta\inv(s)$ intersects both $\eta\inv(r)$ and $\eta\inv(r')$ (otherwise a separator in \bool{\Cs} would be recognized by $\beta$). This exactly says that we have $w,w' \in D^*$ such that $\beta(w) = \alpha(w) = \beta(w') = \alpha(w') = s$, $\eta(w) = r$ and $\eta(w') = r'$, finishing the proof.

  \end{proof}

  We fix the tagging $\eta: D^* \to R$ and $G$ for the remainder of the argument. We define $\ell = |D|$ and $m = |R|$. We may now prove the Lemma~\ref{lem:tags}. We proceed by induction on $h \geq 1$.

  The case $h = 1$ has already been handled with Fact~\ref{fct:themorphism}. Assume now that $h \geq 2$. Induction to $h-1$ yields a morphism $\gamma': (B')^* \to T'$ and $F' \subseteq T'$ satisfying the two assertions in the lemma. Recall that \bool{\Cs} is a \varie by hypothesis. Hence, it is closed under bijective renaming of letters and we may assume without loss of generality that $D \cap B' = \emptyset$. We define the alphabet $B$ as the disjoint union $B = B' \cup D$. Moreover, we let $T$ as the monoid $T = T' \times R$ equipped with the componentwise multiplication. We let $\gamma: B^* \to T$ as the morphism such for every $b \in B$,

  \[
    \gamma(b) =  \left\{
      \begin{array}{ll}
        (\gamma'(b),1_R)     & \text{if $b \in B'$} \\
        (1_{T'},\eta(b)) & \text{if $b \in D$}
      \end{array}
    \right.
  \]

  Finally, we let $F = F' \times G$. Observe that by definition, we have $|F| = 2 \times |F'| \geq 2^{h}$. Moreover, $|B| = |D| + |B'| \leq h \times \ell$ and $|T| = |T'| \times |R| \leq m^h$. It remains to show that the second assertion in Lemma~\ref{lem:tags} holds.

  We consider an alphabet and a morphism $\alpha: (A \cup B)^* \to M$ such that every language recognized by $\alpha$ belong to $\bool{\Cs}(A \cup B)$. We have to exhibit $s \in M$ such for every $t \in F$, there exists $w_t \in B^*$ satisfying $\alpha(w_t) = s$ and $\gamma(w_t) = t$. By hypothesis on $\eta$ and $\gamma'$, we have the following fact.

  \begin{fact}

    We have two elements $s_{B'},s_D \in M$ which satisfy the following properties:

    \begin{itemize}

    \item for every $t' \in F'$, we have $w_{t'} \in (B')^*$ such that $\alpha(w_{t'}) = s_{B'}$ and $\gamma'(w_{t'}) = t'$.

    \item for every $r \in G$, we have $w_{r} \in D^*$ such that $\alpha(w_{r}) = s_{D}$ and $\eta(w_{r}) = r$.

    \end{itemize}

  \end{fact}

  \begin{proof}

    We prove the existence of $s_{B'}$, the argument for $s_D$ is symmetrical. Recall that $B = B' \cup D$ and let $\beta: (A \cup B')^* \to M$ be the restriction of $\alpha$ to $(A \cup B')^*$. Since \bool{\Cs} is a \varie, and all languages recognized by $\alpha$ belong to $\bool{\Cs}(A \cup B)$, it straightforward to verify that all languages recognized by $\beta$ belong to $\bool{\Cs}(A \cup B')$. Hence, since by hypothesis on $\gamma': (B')^* \to T'$ and $F'$, we obtain $s_{B'} \in M$ such that for every $t' \in F'$, we have $w_{t'} \in (B')^*$ such that $\alpha(w_{t'}) = \beta(w_{t'}) = s_{B'}$ and $\gamma'(w_{t'}) = t'$.

  \end{proof}

  We define $s = s_{B'}s_{D}$. It remains to show that $s$ satisfies the desired property. Consider $t \in F = F' \times G$. We have $t = (t',r)$ with $t' \in F'$ and $r \in G$. Let $w_t = w_{t'}w_r$. By definition of $\gamma$, since $w_{t'} \in (B')^*$ and $w_r \in D^*$, we have,

  \[
    \gamma(w_t) = \gamma(w_{t'})\gamma(w_r) = (\gamma'(w_{t'}),1_R) \cdot (1_{T'},\eta(w_r)) = (t',1_R) \cdot (1_{T'},r) = (t',r) = t
  \]
  This concludes the proof.
\end{proof}

\subsection{Proof of Lemma~\ref{lem:extendeasy2}}

We now prove Lemma~\ref{lem:extendeasy2}. Let us first recall the statement.

\adjustc{lem:extendeasy2}

\begin{lemma}
  Let \Cs be a finite \vari. For every $n \in \frac12 \nat$, $\Cs_\frE[n]$ is smooth and an extension of $\Cs[n]$.
\end{lemma}

\restorec

We fix the finite \vari \Cs for the proof. We start by proving that $\Cs_\frE[n]$ is smooth for every $n \in \frac12 \nat$.

Let $k \in \nat$, we describe a tagging of rank $k$. we let $T_k = \{t_0,\dots,t_{k-1}\}$ as the monoid whose multiplication is defined by $t_it_j = t_{i +j \mod k}$ for $i,j \leq k-1$ (i.e. $T$ is isomorphic to \quozk). We now consider the morphism $\tau_k: \frE^* \to T_k$ defined by $\beta(0) = \beta(1) = t_1$ (i.e. $\tau_k$ counts the length of words modulo $k$). Clearly the tagging $(\tau_k: \frE* \to T_k,T_k)$ has rank $k$ and can be computed in \logspace. Moreover, the following lemma can be verified from the definition of $\Cs_\frE$ and that of concatenation hierarchies (the proof is left to the reader).

\begin{lemma} \label{lem:itisfooled}
  For every $k \in \nat$ and every $n \in \frac12 \nat$, the tagging $(\tau_k: \frE* \to T_k,T_k)$ fools $\Cs_\frE[n]$.
\end{lemma}

Altogether, we obtain that $\Cs_\frE[n]$ is smooth for every $n \in \frac12 \nat$. It remains to show that $\Cs_\frE[n]$ is an extension of $\Cs[n]$ for every $n \in \frac12 \nat$. Both conditions involved in extension are verified using induction on $n$ (this amounts to proving that they are preserved by polynomial and Boolean closure). The arguments are straightforward and left to the reader.

\section{Appendix to Section~\ref{sec:fixalph}}
\label{app:fixalph}
In this appendix we present the missing proofs of Section~\ref{sec:fixalph}. Let us first take care of Lemma~\ref{lem:compat}. Recall that in this section, an arbitrary alphabet $A$ and a finite \vari \Cs are fixed.

\subsection{Proof of Lemma~\ref{lem:compat}}

Let us first recall the statement of Lemma~\ref{lem:compat}

\adjustc{lem:compat}
\begin{lemma}
  Given two morphisms recognizing regular languages $L_1,L_2 \subseteq A^*$ as input, one may compute in \logspace a \Cs-compatible morphism which recognizes both $L_1$ and $L_2$.
\end{lemma}
\restorec

We let $\alpha_1: A^* \to M_1$ and $\alpha_2: A^* \to M_2$ as the morphisms recognizing $L_1$ and $L_2$. Recall that the relation $\sim_\Cs$ associated to \Cs is a congruence over $A^*$ for word concatenation ($\sim_\Cs$ compares words which belong to the same languages in \Cs). Therefore, the quotient set ${A^*}/{\sim_\Cs}$ is a monoid (we write ``\cmult'' for its multiplication) and the map $w \mapsto \ctype{w}$ which maps each word to its $\sim_\Cs$-class is a monoid morphism.

We let $M = M_1 \times M_2 \times ({A^*}/{\sim_\Cs})$ as the monoid equipped with the componentwise multiplication. Moreover, we let $\beta: A^* \to M$ as the morphism defined by $\beta(w) = (\alpha_1(w),\alpha_2(w),\ctype{w})$. Clearly, $\beta$ recognizes both $L_1$ and $L_2$. Moreover, $\beta$ is \Cs-compatible: given $s = (s_1,s_2,D) \in M$, it suffices to define $\ctype{s} = D$. It then immediate that the two axioms in the definition of \Cs-compatibility are satisfied:
\begin{itemize}
\item Given $w \in A^*$ we $\ctype{\beta(w) } = \ctype{w}$.
\item Given $s,s' \in M$ $\ctype{ss'} = \ctype{s} \cmult \ctype{s'}$.
\end{itemize}
Finally, it is clear that $\beta$ ca be computed in \logspace from $\alpha_1$ and $\alpha_2$.

\begin{remark}
  It is important here that the alphabet $A$ is fixed. This implies that the monoid ${A^*}/{\sim_\Cs}$ is a constant. When $A$ is a parameter, it may not be possible to compute $\beta$ in \logspace (this depends on \Cs).
\end{remark}

\subsection{Proof of Proposition~\ref{prop:computimp}}

We actually prove a statement which is slightly stronger than Proposition~\ref{prop:computimp} (this is required to use induction in the proof). It is as follows.

\begin{proposition} \label{prop:computimp2}
  Let $h,m \in \nat$ be constants. Consider two \Cs-compatible morphisms $\alpha: A^* \to M$ and $\beta: A^* \to N$ and a good subset $S \subseteq N$. Given $s \in M$ and $T \in 2^N$ such that $|T| \leq m$, one may test in \nlog with respect to $|M|$ and $|N|$ whether there exists an $(\alpha,\beta,S)$-tree  of operational height at most h and with root label $(s,T)$.
\end{proposition}

Clearly, Proposition~\ref{prop:computimp} is the special case of Proposition~\ref{prop:computimp2} when $m = 1$. Hence, we may concentrate on proving  Proposition~\ref{prop:computimp2}.

Consider two \Cs-compatible morphisms $\alpha: A^* \to M$ and $\beta: A^* \to N$ and a good subset $S \subseteq N$. Given $h,m \in \nat$, we shall write $X_{h,m} \subseteq M \times 2^N$ for the set of all elements $(s,T) \in M \times 2^N$ such that $|T| \leq m$ and $(s,T)$ is the root label of an $(\alpha,\beta,S)$-tree of operational height is a most $h$.

We have to show that when $h$ and $m$ are fixed, one may test in \nlog with respect to $|M|$ and $|N|$ whether some input pair $(s,T) \in M \times 2^N$ belongs to $X_{h,m}$. We proceed by induction on $h$.

\medskip

When $h = 0$, $(\alpha,\beta,S)$-trees of operational height $0$ contain only leaves and binary nodes. Therefore, one may verify from the definition that their labels are always of the form $(\alpha(w),\{\beta(w)\})$ for some $w \in A^*$. Consequently, the problem of deciding whether $(s,T)$ belongs to $X_{h,m}$ amounts to verifying that $T$ is a singleton $\{t\}$ and that there exists $w \in A^*$ such that $\alpha(w) = s$ and $\beta(w) = t$. This is easily achieved in \nlog.

\medskip

We now assume that $h \geq 1$. We introduce an auxiliary set $Y_{h,m} \subseteq M \times 2^N$. Given $(s,T) \in M \times 2^N$, we have $(s,T) \in Y_{h,m}$ when $|T| \leq m$ and one of the two following conditions holds:
\begin{itemize}
\item $(s,T) \in X_{h-1,m}$, or,
\item $(s,T)$ is the root label of an $(\alpha,\beta,S)$-tree having operational height $h$ and whose root is an $S$-operation node (i.e. the unique child of the root has operational height $h-1$).
\end{itemize}
By induction on $h$, we have the following lemma.

\begin{lemma} \label{lem:induc}
  Let $s \in M$ and $T \in 2^N$, one may test in \nlog with respect to $|M|$ and $|N|$ whether $(s,T) \in Y_{h,m}$
\end{lemma}

\begin{proof}
  It suffices to verify that given as input $(s,T) \in Y_{h,m}$ such that $|T| \leq m$, one may check in \nlog whether one of the two conditions in the definition of $Y_{h,m}$ is satisfied. Testing whether $(s,T) \in X_{h-1,m}$ can be achieved in \nlog by induction on $h-1$. For the second condition, we know that the two following properties are equivalent:
  \begin{itemize}
  \item $(s,T)$ is the root label of an $(\alpha,\beta,S)$-tree having operational height at $h$ and whose root is an $S$-operation node.
  \item there exists an $(\alpha,\beta,S)$-tree having operational height $h-1$ whose root label $(e,E)$ is an idempotent satisfying:
    \[
      e =s \quad \text{and} \quad T \subseteq E \cdot \{t \in S \mid \ctype{e} = \ctype{t} \in S\} \cdot E
    \]
  \end{itemize}
  Since $|T| \leq m$, it is straightforward to verify that the second assertion is satisfied if and only if $E$ can be chosen such that $|E| \leq 2m$ (i.e. $(e,E) \in X_{h-1,2m}$). Hence, the second conditions can be checked in \nlog by induction which concludes the proof.
\end{proof}

Moreover, the next lemma is immediate from the definition of $(\alpha,\beta,S)$-trees of operational height $h$ and a pigeon-hole principle argument.

\begin{lemma} \label{lem:bound}
  Let $(s,T) \in M \times 2^N$. Then, $(s,T) \in X_{h,m}$ if and only if there exists $\ell \leq |M| \times |N|^m$ and $\ell$ elements $(r_1,T_1), \dots,(r_\ell,T_\ell) \in Y_{h,m}$ such that,
  \[
    s = r_1 \cdots r_\ell \quad \text{and} \quad \{t_1,\dots,t_m\} \subseteq T_1 \cdots T_\ell
  \]
\end{lemma}

It is now immediate from Lemma~\ref{lem:induc} and~\ref{lem:bound} that one may test in \nlog with respect to $|M|$ and $|N|$ whether some input pair $(s,T) \in M \times 2^N$ belongs to $X_{h,m}$. This concludes the proof.

\subsection{Proof of Proposition~\ref{prop:opbound}}

Let us first recall the statement of Proposition~\ref{prop:opbound}.

\adjustc{prop:opbound}
\begin{proposition}
  Let $h \in \nat$ be the \Js-depth of ${A^*}/{\sim_\Cs}$. Consider two \Cs-compatible morphisms $\alpha: A^* \to M$ and $\beta: A^* \to N$, and a good subset $S \subseteq N$. Then, for every $(s,T) \in M \times 2^N$, the following properties are equivalent:
  \begin{enumerate}
  \item $(s,T)$ is the root label of some $(\alpha,\beta,S)$-tree.
  \item $(s,T)$ is the root label of some $(\alpha,\beta,S)$-tree whose operational height is at most $h$.
  \end{enumerate}
\end{proposition}
\restorec

We fix $h$ as the \Js-depth of ${A^*}/{\sim_\Cs}$. Moreover, we let $\alpha: A^* \to M$ and $\beta: A^* \to N$ as two \Cs-compatible morphisms and fix $S\subseteq N$ as a good subset. The direction $2)\Rightarrow 1)$ in Proposition~\ref{prop:opbound} is trivial. Therefore, we concentrate on proving that $1) \Rightarrow 2)$. Given $(s,T) \in M \times 2^N$ and a $(\alpha,\beta,S)$-tree \frT whose root label is $(s,T)$, we explain how to construct a second tree with the same root label and whose operational height is bounded by $h$.

For the proof, we call \emph{operational size} of an $(\alpha,\beta,S)$-tree the total number of operation nodes it contains (clearly, this number is always larger than the operational height). The result is a consequence of the following lemma.

\begin{lemma} \label{lem:half:pumpingtree}
  Consider an $(\alpha,\beta,S)$-tree \frT and assume that it contains a branch with two distinct operation nodes $x$ and $x'$ whose labels $(s,T)$ and $(s',T')$ satisfy $\ctype{s} = \ctype{s'}$. Then, there exists a second tree $\frT'$ with strictly smaller operational size than \frT and with the same root label.
\end{lemma}

Starting from an arbitrary $(\alpha,\beta,S)$-tree \frT, one may use Lemma~\ref{lem:half:pumpingtree} recursively to build $\frT'$ which has the same label as \frT and such that for any two operation nodes $x$ and $x'$ on the same branch of $\frT'$, their labels $(s,T)$ and $(s',T')$ satisfy $\ctype{s} \neq \ctype{s'}$. Clearly, this tree $\frT'$ has operational height bounded by $h$ (by definition of $h$ as the \Js-depth of ${A^*}/{\sim_\Cs}$). This concludes the proof for the implication $1) \Rightarrow 2)$ in Proposition~\ref{prop:opbound}.

\medskip

We now concentrate on proving Lemma~\ref{lem:half:pumpingtree}. We let \frT and $x \neq x'$ the nodes defined in the lemma. Since $x,x'$ are on the same branch, one is an ancestor of the other. By symmetry, we assume that $x$ is an ancestor of $x'$. We let \frS as the subtree of \frT which is rooted in $x$. We let $(s,T)$ as the label $(s,T) = lab(\frS) = lab(x)$. We build a new tree $\frS'$ with the same label as \frS and strictly smaller operational size. It will then be simple to build the desired tree $\frT'$ by replacing the subtree \frS with $\frS'$ in \frT.

Given two nodes $z,z'$ of \frS, we write $z < z'$ to denote the fact that $z$ is a (strict) ancestor of $z'$. By hypothesis, we have $x < x'$, hence we may consider the sequence of operations nodes which are between the two. We let $x_1,\dots,x_k$ as the sequence of all nodes which satisfy the following properties:
\begin{itemize}
\item For all $i$, $x_i$ is an operation node.
\item $x = x_k < \cdots < x_1 = x'$.
\end{itemize}
Note that since $x_k = x$ and $x_1 = x'$, we have $k \geq 2$. For all $i \geq 1$, we let $(f_i,T_i)$ as label of $x_i$. By definition of operation nodes, $f_i \in M$ must be an idempotent. Moreover, $(f_k,T_k) = (s,T)$ is the label of \frS and we know by hypothesis that $\ctype{f_1} = \ctype{f_k}$. Finally, consider the unique child of $x_1$ and let $(e,E)$ be the label of this child (which is an idempotent of $M \times 2^N$ since $x_1$ is an operation node). Recall that by definition of operation nodes, we have $e = f_1$ and $T_1 \subseteq E \cdot \{t \in S \mid \ctype{e} = \ctype{t}\} \cdot E$.

We now classify the nodes within \frS in several categories. We call \emph{backbone} of \frS the path made of all (strict) ancestors of $x_1$. Since $x_k$ is the root, there are $k-1 \geq 1$ operation nodes on the backbone (the nodes $x_2,\dots,x_k$). Furthermore, we call \emph{lower nodes} all nodes within the subtree rooted in $x_1$ (including $x_1$). We denote by $m$ the number operation nodes which are lower nodes. Finally, all nodes which are neither backbone nor lower nodes are called \emph{side nodes}. Observe that any side node $z$ has a closest ancestor $y$ on the backbone which has to be a binary node. We say that $z$ is a \emph{left (resp. right) side node} when it belongs to the subtree whose root is the left (resp. right) child of $y$. Finally, we associate a \emph{rank} to each side node $z$: the rank of $z$ is the smallest $i \leq k$ such that $x_i$ is an ancestor of $z$ ($i$ must exist since $x_k$ is the root). For all $i \leq k$, we write $\ell_i$ (resp. $r_i$) the number of operation nodes which are left (resp. right) side nodes of rank $i$. We illustrate these definitions in Figure~\ref{fig:half:classnodes}.

\begin{figure}
  \begin{center}
    \begin{tikzpicture}
      \node[snode,fill=red!30] (xn) at (0.0,6.0) {$x_k$};

      \node[snode,fill=red!30] (x3) at (0.0,4.5) {$x_3$};

      \node[snode,fill=red!30] (x2) at (0.0,1.5) {$x_2$};

      \node[snode,fill=red!30] (x1) at (0.0,-0.2) {$x_1$};

      \node[snode,fill=blue!30] (y2) at (0.0,0.75) {};

      \node[snode,fill=blue!30] (z1) at (0.0,2.25) {};
      \node[snode,fill=blue!30] (z2) at (0.0,3.0) {};
      \node[snode,fill=blue!30] (z3) at (0.0,3.75) {};

      \draw[lines] (x1) to (y2);
      \draw[lines] (y2) to (x2);

      \draw[lines] (x2) to (z1);
      \draw[lines] (z1) to (z2);
      \draw[lines] (z2) to (z3);
      \draw[lines] (z3) to (x3);

      \draw[lines,dotted] (x3) to (xn);

      \node[triangle,minimum height=1.7cm] (t1) at (0,-1.4) {};
      \draw[lines] (x1) to ($(t1.corner 1)-(0,0.10)$);

      \node[triangle,minimum height=1.1cm] (t2) at (-2,-0.15) {};
      \draw[lines] (y2) to ($(t2.corner 1)-(0,0.1)$);

      \node[triangle,minimum height=1.1cm] (t3) at (2,1.35) {};
      \draw[lines] (z1) to ($(t3.corner 1)-(0,0.1)$);

      \node[triangle,minimum height=1.1cm] (t4) at (-2,2.1) {};
      \draw[lines] (z2) to ($(t4.corner 1)-(0,0.1)$);

      \node[triangle,minimum height=1.1cm] (t5) at (2,2.85) {};
      \draw[lines] (z3) to ($(t5.corner 1)-(0,0.1)$);

      \begin{pgfonlayer}{background}

        \node[bobox,fit=(x1) (t1.corner 2) (t1.corner 3)] (lower) {};

        \node[bobox,fit=(xn) (y2)] (backbone) {};

        \node[bobox,fit=(t2.corner 1) (t2.corner 2) (t2.corner 3)] (leftrankone) {};

        \node[bobox,fit=(t4.corner 1) (t4.corner 2) (t4.corner 3)] (leftranktwo) {};
        \node[bobox,fit=(t5.corner 1) (t3.corner 2) (t3.corner 3)] (rightranktwo) {};

      \end{pgfonlayer}

      \node at ($(lower.south)-(0.0,0.25)$) {Lower nodes};
      \node[rotate=90] at ($(backbone.102)-(0.25,0.0)$) {Backbone};

      \node[rotate=90,align=center] at ($(leftrankone.west)-(0.5,0.0)$) {Left side nodes\\ of rank $2$};

      \node[rotate=90,align=center] at ($(leftranktwo.west)-(0.5,0.0)$) {Left side nodes\\ of rank $3$};

      \node[rotate=90,align=center] at ($(rightranktwo.east)+(0.5,0.0)$) {Right side nodes\\ of rank $3$};

      \node[snode,fill=red!30] (l1) at (2.2,-0.7) {};
      \node[snode,fill=blue!30] (l2) at (2.2,-1.3) {};

      \node[anchor=mid west] at ($(l1)+(0.25,0.0)$) {Operation};
      \node[anchor=mid west] at ($(l2)+(0.25,0.0)$) {Binary};

      \draw[very thick,decorate,decoration=brace] ($(l2)-(0.3,0.35)$) to ($(l1)-(0.3,-0.35)$);

    \end{tikzpicture}
  \end{center}
  \caption{Classification of the nodes in \frS (here, there are no right side nodes of rank $2$).}
  \label{fig:half:classnodes}
\end{figure}
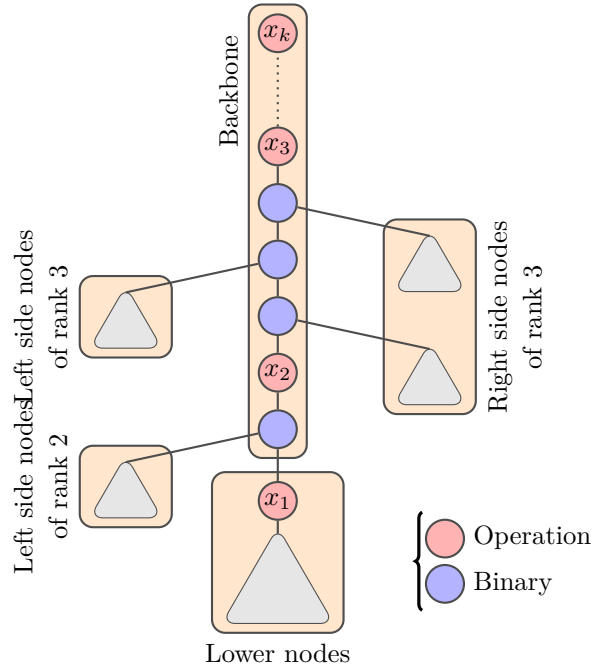

Observe that by definition, backbone nodes, lower nodes and side nodes account for all nodes in the tree. Thus, we have the following fact.

\begin{fact} \label{fct:half:totoalcount}
  The total number of operation nodes in \frS is,
  \[
    k-1 + m + \ell_1 + \cdots \ell_k + r_1 + \cdots + r_k
  \]
\end{fact}

Essentially, the desired tree $\frS'$ is built by removing all backbone nodes from \frS and replacing them with binary nodes. Thus, we obtain a tree $\frS'$ whose operational size is $m + \ell_1 + \cdots \ell_k + r_1 + \cdots + r_k$ which is strictly smaller than that of \frS since $k-1 \geq 1$. We use an inductive construction which is formalized in the following lemma.

\begin{lemma} \label{lem:half:crunchtree}
  For every $i \leq k$, there exist two $(\alpha,\beta,S)$-trees $\frU_i$ and $\frV_i$ of labels $(u_i,U_i)$ and $(v_i,V_i)$ with operational heights $\ell_1 + \cdots + \ell_i$ and $r_1 + \cdots + r_i$ respectively. Moreover, there exist $u'_i,v'_i \in M$ satisfying the following two conditions:
  \begin{enumerate}
  \item For $q \in \{u_i,u'_i\}$ and $r \in \{v_i,v'_i\}$, $f_i  = qer$.
  \item $T_i \subseteq U_i E \cdot \{t \in S \mid \ctype{t} = \ctype{ev'_if_iu'_ie}\} \cdot E V_i$.
  \end{enumerate}
\end{lemma}

Before we show Lemma~\ref{lem:half:crunchtree}, we use it to build the desired tree $\frS'$ and finish the proof of Lemma~\ref{lem:half:pumpingtree}. Recall that we need $\frS'$ to have label $lab(\frS) = (s,T) = (f_k,T_k)$. We apply Lemma~\ref{lem:half:crunchtree} in the special case when $i = k$. This yields two $(\alpha,\beta,S)$-trees $\frU_k$ and $\frV_k$ with labels $(u_k,U_k)$ and $(v_k,V_k)$ which have operational heights $\ell_1 + \cdots + \ell_i$ and $r_1 + \cdots + r_i$. Moreover, we let $u'_k,v'_k \in M$ which satisfy the two assertions in the lemma.

It follows from the first assertion in Lemma~\ref{lem:half:crunchtree} that $u_kev_k = v'_keu'_k = f_k = s$. This implies the following fact.

\begin{fact} \label{fct:half:somealgebra}
  $\ctype{e} = \ctype{ev'_kf_ku'_k e}$.
\end{fact}

\begin{proof}
  By definition of \Cs-compatible morphisms we have,
  \[
    \ctype{ev'_kf_ku'_k e} = \ctype{e} \cmult \ctype{v'_k} \cmult \ctype{f_k} \cmult \ctype{u'_k} \cmult\ctype{e}
  \]
  Therefore, since $\ctype{f_k} = \ctype{e}$, it suffices to prove that,	$\ctype{e} = \ctype{e} \cmult \ctype{v'_k} \cmult\ctype{e}\cmult \ctype{u'_k} \cmult\ctype{e}$.

  By the first assertion in Lemma~\ref{lem:half:crunchtree}, we have $e = f_k  =u'_kev'_k$. Hence, $\ctype{e} = \ctype{u'_k} \cmult \ctype{e} \cmult \ctype{v'_k}$. Moreover, since $e$ is idempotent of $M$, $\ctype{e} = \ctype{ee} = \ctype{e} \cmult \ctype{e}$ is an idempotent of ${A^*}/{\sim_\Cs}$. This yields,
  \[
    \begin{array}{lll}
      \ctype{e} & = & \ctype{e} \cmult \ctype{u'_k} \cmult \ctype{e} \cmult \ctype{v'_k} \cmult\ctype{e}\\
      \ctype{e} & = & (\ctype{e} \cmult \ctype{u'_k})^\omega \cmult \ctype{e} \cmult (\ctype{v'_k} \cmult\ctype{e})^\omega\\
      \ctype{e} & = & \ctype{e} \cmult (\ctype{v'_k} \cmult\ctype{e})^\omega\\
      \ctype{e} & = & \ctype{e} \cmult \ctype{v'_k} \cmult\ctype{e} \cmult (\ctype{v'_k} \cmult\ctype{e})^{\omega-1}
    \end{array}
  \]
  We may now replace the second copy of $\ctype{e}$ in the above with $\ctype{e} \cmult \ctype{u'_k} \cmult \ctype{e} \cmult \ctype{v'_k} \cmult\ctype{e}$ which yields,
  \[
    \ctype{e}  = \ctype{e} \cmult \ctype{v'_k} \cmult\ctype{e} \cmult\ctype{u'_k} \cmult \ctype{e} \cmult (\ctype{v'_k} \cmult\ctype{e})^{\omega}
  \]
  Finally, since $\ctype{e} = \ctype{e} \cmult (\ctype{v'_k} \cmult\ctype{e})^\omega$, this yields $\ctype{e} = \ctype{e} \cmult \ctype{v'_k} \cmult\ctype{e}\cmult \ctype{u'_k} \cmult \ctype{e}$ as desired.
\end{proof}

In view of Fact~\ref{fct:half:somealgebra} and the second assertion in Lemma~\ref{lem:half:crunchtree}, we obtain that,
\begin{equation} \label{eq:tkinc}
  T_k \subseteq U_k E \cdot \{t \in S \mid \ctype{t} = \ctype{e}\} \cdot E V_k
\end{equation}
Finally, we have a tree of root label $(e,E)$ whose operational size is $m-1$: the child of $x_1$. Hence, using one operation node, we may build a tree of operational size $m$ whose root label is:
\[
  (e,E \cdot \{t \in S \mid \ctype{t} = \ctype{e}\})
\]
Finally, by~\eqref{eq:tkinc}, we may combine this tree with $\frU_k$ and $\frV_k$ using two binary nodes to get a tree $\frS'$ whose root label is:
\[
  (s,T) = (f_k,T_k) = (u_kev_k,T_k)
\]
By definition, this tree $\frS'$ has operational size $m + m +\ell_1 + \cdots + \ell_k+r_1 + \cdots + r_k$. As desired, this is strictly smaller than \frS (its operational size is $k-1 + m + \ell_1 + \cdots \ell_k + r_1 + \cdots + r_k$ by Fact~\ref{fct:half:totoalcount} and $k-1 \geq 1$). This terminates the proof of Lemma~\ref{lem:half:pumpingtree}.

\medskip

It now remains to prove Lemma~\ref{lem:half:crunchtree}. We proceed by induction on $i$. When $i = 1$, since $x_1$ is an operation node whose unique child has label $(e,E)$, we have $f_1 = e$ and $T_1 \subseteq E \cdot \{t \in S \mid \ctype{e} = \ctype{t}\} \cdot E$. We define both $\frU_1$ and $\frV_1$ as the same tree made of a single leaf whose label is $(1_M,\{1_N\}) = (\alpha(\varepsilon),\{\beta(\varepsilon)\})$. It is then simple to verify that the two assertions in the lemma are satisfied for $u'_1 = v'_1 = 1_M$.

\medskip

We now assume that $i \geq 2$. By definition, $x_i$ has a unique child whose label is an idempotent $(f_i,F_i)$ such that,
\[
  T_i \subseteq F_i \cdot \{t \in S \mid \ctype{f_i} = \ctype{t}\} \cdot F_i
\]
We use the following fact to choose our new trees $\frU_i,\frV_i$.

\begin{fact} \label{fct:half:takesides}
  There exist two $(\alpha,\beta,S)$-trees $\frP$ and $\frQ$ whose operational sizes are respectively bounded by $\ell_i$ and $r_i$ and whose labels $(p,P)$ and $(q,Q)$ satisfy the following two properties,
  \begin{itemize}
  \item $f_i = p \cdot f_{i-1} \cdot q$
  \item $F_i \subseteq P T_{i-1} Q$
  \end{itemize}
\end{fact}

\begin{proof}
  We build \frP (resp. \frQ) by combining all subtrees made of left (resp. right) side nodes of rank $i$ into a single one using binary nodes only. In the degenerate case when there are no left (resp. right) side nodes \frP (resp. \frQ) is a single leaf with label $(1_M,\{1_N\})$.

  Let us describe this construction in more details when the set of left and right side nodes of rank $i$ are nonempty Consider all nodes between $x_i$ and $x_{i-1}$ (which are all binary by definition). For each such node, one child is an ancestor of $x_{i-1}$ (or $x_{i-1}$ itself) and the other is a side node. We define,
  \begin{itemize}
  \item $x_i < z_{h_1} < \cdots < z_1 < x_{i-1}$ as all binary nodes whose left children are side nodes (in particular these children and all their descendants are left side nodes of rank $i$).
  \item $x_i < z'_{h_2} < \cdots < z'_1 < x_{i-1}$ as all binary nodes whose right children are side nodes (in particular these children and all their descendants are right side nodes of rank $i$).
  \end{itemize}
  We may now define \frP and \frQ. We start with \frP. For all $j \leq h_1$, we let $(p_j,P_j)$ as the label of the left child of $z_j$. Clearly, one may combine all subtrees rooted in the left children of the $z_j$ with binary nodes into a single one whose label is,
  \[
    (p,P) = (p_{h_1},P_{h_1}) \cdot \cdots \cdot (p_{1},P_{1})
  \]
  By definition, the operational size of \frP is $\ell_i$: the sum of those for the subtrees we have combined (we only added binary nodes). Symmetrically, one may build \frQ of operational size $r_i$ whose label is,
  \[
    (q,Q) = (q_{1},Q_{1}) \cdot \cdots \cdot (q_{h_2},Q_{h_2})
  \]
  where $(q_j,Q_j)$ is the label of the right child of $z'_j$ for all $j \leq h_2$. One may now verify from the definition that the two assertions in the fact are satisfied.
\end{proof}

We are now ready to define our new trees $\frU_i$ and $\frV_i$. We first use induction to obtain two trees $\frU_{i-1}$ and $\frV_{i-1}$ of labels $(u_{i-1},U_{i-1})$ and $(v_{i-1},V_{i-1})$ which satisfy the conditions of Lemma~\ref{lem:half:crunchtree} for $i-1$. We define,
\begin{itemize}
\item $\frU_i$ as the tree of label $(u_i,U_i) = (p \cdot u_{i-1},PU_{i-1})$ obtained by combining \frP and $\frU_{i-1}$ with a single binary node.
\item $\frV_i$ as the tree of label $(v_i,V_i) = (v_{i-1} \cdot q,V_{i-1}S)$ obtained by combining $\frV_{i-1}$ and \frQ with a single binary node.
\end{itemize}

It remains to prove that this definition for the trees $\frU_i$ and $\frV_i$ satisfies the conditions in Lemma~\ref{lem:half:crunchtree}. By definition, the operational size of $\frU_i$ is the sum of that of \frP (i.e. $\ell_i$ by definition in Fact~\ref{fct:half:takesides}) with that of $\frU_{i-1}$ (i.e. $\ell_{1}+\cdots \ell_{i-1}$ since we obtained $\frU_{i-1}$ by induction). This exactly says that the operational size of $\frU_i$ is $\ell_{1}+\cdots \ell_i$ as desired. Symmetrically, one may verify that the operational size of $\frV_i$ is $r_1+\cdots +r_i$.

\medskip

We now have to find $u'_i,v'_i \in M$ which satisfy the two assertions in the lemma. Since we obtained $\frU_{i-1}$ and $\frV_{i-1}$ by induction, we also have $u'_{i-1},v'_{i-1} \in \Lb$ which satisfy these two assertions for $i-1$. We define,
\[
  u'_i  = p f_{i-1} u'_{i-1} \quad \text{and} \quad v'_i = v'_{i-1} f_{i-1} q
\]

It remains to verify that the two assertions in Lemma~\ref{lem:half:crunchtree} hold for this choice of $u'_i,v'_i$. We begin with the first one.

\medskip
\noindent
{\bf Assertion~1.} We have four equalities to verify. Since the argument is similar for all four, we concentrate on $f_i  = u_iev_i$ and $f_i  = u'_iev'_i$ whose proofs encompass all arguments. By Fact~\ref{fct:half:takesides}, we know that $f_i = pf_{i-1} q$. Moreover, since $f_{i-1}  = u_{i-1}ev_{i-1}$ by the inductive definition of $u_{i-1}$ and $v_{i-1}$, we get,
\[
  f_i =  pu_{i-1}ev_{i-1}q =  u_iev_i
\]
Furthermore, $f_{i-1}$ is idempotent. Thus, $f_i =  pf_{i-1}q = p (f_{i-1})^3q$ and since by construction of $u'_{i-1}$ and $v'_{i-1}$, we have $f_{i-1}  = u'_{i-1}ev'_{i-1}$, we obtain,
\[
  f_i =  pf_{i-1} u'_{i-1} e v'_{i-1} f_{i-1} q =  u'_i e v'_i
\]

\medskip
\noindent
{\bf Assertion~2.} We finish with the second assertion which is the most involved. In particular, this is where we use the fact that $S$ is good. We need to show that,
\[
  T_i \subseteq U_i E \cdot \{t \in S \mid \ctype{t} = \ctype{ev'_if_iu'_ie}\} \cdot E V_i
\]
We start with a simple fact.

\begin{fact} \label{fct:half:correctype}
  For any $(s,T) \in M \times 2^N$ which is the label of an $(\alpha,\beta,S)$-tree, we have $T \subseteq \{t \in S \mid \ctype{t} = \ctype{s}\}$.
\end{fact}

\begin{proof}
  This is immediate by induction on the height of $(\alpha,\beta,S)$-trees using the hypothesis that $S$ is good.
\end{proof}

We now start the proof. By definition, $(f_i,T_i)$ is the label of the operation node $x_i$ whose child has label $(f_i,F_i)$. Hence, $T_i \subseteq F_i \cdot  \{t \in S \mid \ctype{t} = \ctype{f_i}\} \cdot F_i$ and it follows from the second item in Fact~\ref{fct:half:takesides} that,
\[
  T_i \subseteq P T_{i-1} Q \cdot \{t \in S \mid \ctype{t} = \ctype{f_i}\}  \cdot P T_{i-1} Q
\]
The result is now a consequence of the two following inclusions:
\begin{equation} \label{eq:half:biginclusion}
  \begin{array}{rll}
    P T_{i-1} Q & \subseteq & U_i E \cdot \{t \in S \mid \ctype{t} = \ctype{ev'_i}\} \\
    P T_{i-1} Q & \subseteq & \{t \in S \mid \ctype{t} = \ctype{u'_ie}\} \cdot EV_i
  \end{array}
\end{equation}
Indeed, one may combine these two inequalities with the previous one using the hypothesis that $S$ is good to obtain the desired inclusion:
\[
  \begin{array}{lll}
    T_i & \subseteq & U_i E \cdot \{t \in S \mid \ctype{t} = \ctype{ev'_i}\} \cdot \{t \in S \mid \ctype{t} = \ctype{f_i}\}  \cdot \{t \in S \mid \ctype{t} = \ctype{u'_ie}\} \cdot EV_i \\
        & \subseteq & U_i E \cdot \{t \in S \mid \ctype{t} = \ctype{ev'_if_iu'_ie}\} \cdot EV_i
  \end{array}
\]
It remains to prove the two inequalities in~\eqref{eq:half:biginclusion}. As they are based on symmetrical arguments, we concentrate on the first one and leave the other to the reader. Since we built $U_{i-1}$ and $V_{i-1}$ with induction, we have,
\[
  T_{i-1} \subseteq U_{i-1} E \cdot \{t \in S \mid \ctype{t} = \ctype{ev'_{i-1}f_{i-1}u'_{i-1}e}\} \cdot EV_{i-1}
\]
By Fact~\ref{fct:half:correctype}, $E \subseteq \{t \in S \mid \ctype{t} = \ctype{e}\}$ and $V_{i-1} \subseteq \{t \in S \mid \ctype{t} = \ctype{v_{i-1}}\}$. Hence, using the fact that $S$ is good, we may simplify the above inclusion as follows:
\[
  T_{i-1} \subseteq U_{i-1} E \cdot \{t \in S \mid \ctype{t} = \ctype{ev'_{i-1}f_{i-1}u'_{i-1}ev_{i-1}}\}
\]
Since $u'_{i-1}$ and $v_{i-1}$ were built by induction, we know that $u'_{i-1}ev_{i-1} = f_{i-1}$. Hence, since $f_{i-1}$ is an idempotent,
\[
  T_{i-1} \subseteq U_{i-1} E \cdot \{t \in S \mid \ctype{t} = \ctype{ev'_{i-1}f_{i-1}}\}
\]
Using Fact~\ref{fct:half:correctype} again, we have $Q \subseteq \{t \in S \mid \ctype{t} = \ctype{q}\}$. Thus, using the hypothesis that $S$ is good together with the fact that $v'_i = v'_{i-1}f_{i-1}q$ by definition, this yields the following,
\[
  \begin{array}{lll}
    T_{i-1} Q & \subseteq & U_{i-1} E \cdot \{t \in S \mid \ctype{t} = \ctype{ev'_{i-1}f_{i-1}q}\} \\
              & \subseteq & U_{i-1} E \cdot \{t \in S \mid \ctype{t} = \ctype{ev'_{i}}\}
  \end{array}
\]
Finally, since $U_i = P U_{i-1}$ by definition, we have
\[
  \begin{array}{lll}
    P T_{i-1} Q & \subseteq & PU_{i-1} E \cdot \{t \in S \mid \ctype{t} = \ctype{ev'_{i}}\}        \\
                & \subseteq & U_i E \cdot \{t \in S \mid \ctype{t} = \ctype{ev'_{i}}\}
  \end{array}
\]
This conclude the proof of Lemma~\ref{lem:half:crunchtree}.

\section{Appendix to Section~\ref{sec:classic}}
\label{app:classic}
This section provides the missing proofs in Section~\ref{sec:classic}. We start by introducing additional terminology and preliminary results that we shall need to present these proofs.

\subsection{Stratifications}

We present a stratification of $\sthtwo = \pol{\at}$ into finite \pvaris. It was introduced in~\cite{pzbpol}. We refer the reader to~\cite{pzbpol} for the proofs of the statements presented here.

For any natural number $k \in \nat$, we define a finite \pvari $\polk{\at} \subseteq \pol{\at}$. The definition uses induction on $k$:

\begin{itemize}

\item When $k = 0$, we simply define $\polp{\at}{0} = \at$.

\item When $k \geq 1$, we define \polk{\at} as the smallest lattice which contains \polp{\at}{k-1} and such for any $L_1,L_2 \in \polp{\at}{k-1}$ and any $a \in A$,

  \[
    L_1 a L_2 \in \polk{\at}
  \]

\end{itemize}

One may verify from the definitions that for every $k \in \nat$, $\polk{\at}$ is a finite \pvari and that $\polp{\at}{k} \subseteq \polp{\at}{k+1}$. Moreover, by definition of \pol{\at}, we have:
\[
  \sthtwo = \pol{\at} = \bigcup_{k \geq 0} \polk{\at}.
\]

Given any alphabet $A$, we associate preorder relations to the strata \polk{\at}. For every $k \in \nat$ and $u,v \in A^*$, we write $u \polrelk v$ when the following condition is satisfied,
\[
  \text{For every $L \in \polk{\at}(A)$,} \quad u \in L \Rightarrow v \in L
\]

It is immediate by definition that \polrelk is a preorder relation on $A^*$. The key point is that we may use it to characterize separability for $\pol{\at} = \sthtwo$.

\begin{lemma} \label{lem:sthsep}
  Let $A$ be an alphabet and $L,L' \subseteq A^*$ two languages. Then, the two following properties are equivalent:
  \begin{enumerate}
  \item $L$ is \textbf{not} \sthtwo-separable from $L'$.
  \item For every $k \in \nat$, there exists $w \in L$ and $w' \in L'$ such that $w \polrelk w'$.
  \end{enumerate}
\end{lemma}

Moreover, we may also use \polrelk to characterize separability for $\bpol{\at} = \sttwo$.

\begin{lemma} \label{lem:stsep}
  Let $A$ be an alphabet and $L,L' \subseteq A^*$ two languages. Then, the two following properties are equivalent:
  \begin{enumerate}
  \item $L$ is \textbf{not} \sttwo-separable from $L'$.
  \item For every $k \in \nat$, there exists $w \in L$ and $w' \in L'$ such that $w \polrelk w'$ and $w' \polrelk w$.
  \end{enumerate}
\end{lemma}

We finish the presentation with three properties of the relations \polrelk. The first one is simple and states that they are compatible with word (this is because the strata \polk{\at} are closed under quotients).

\begin{lemma} \label{lem:concat}
  Let $A$ be an alphabet and $k \in \nat$. For every  $u_1,u_2;v_1,v_2 \in A^*$ such that $u_1 \polrelk v_1$ and $u_2 \polrelk v_2$, we have $u_1u_2 \polrelk v_1v_2$.
\end{lemma}

The second lemma holds because \pol{\at} is a sub-class of the star-free languages. It is as follows.
\begin{lemma} \label{lem:propreo1}
  Let $A$ be an alphabet and $k \in \nat$. Consider $h_1,h_2 \geq 3^{k+1}-1$ and any $u \in A^*$. Then, we have $u^{h_1} \polrelk u^{h_2}$.
\end{lemma}

Finally, the third lemma states a characteristic property of \pol{\at}. The proof is rather technical (see~\cite{pzbpol} for details). Given an alphabet $A$ and a word $w \in A^*$, we write $\cont{w}$ for the \emph{alphabet} of $w$, i.e. the least sub-alphabet $B \subseteq A$ such $w \in B^*$.

\begin{lemma} \label{lem:propreo2}
  Let $A$ be an alphabet and $k \in \nat$. Consider $h,h_1,h_2 \geq 3^{k+1}-1$ and any $u,v \in A^*$ such that $\cont{v} \subseteq \cont{u}$, we have $u^{h} \polrelk u^{h_1} v u^{h_2}$.
\end{lemma}

\subsection{Upper bound in Theorem~\ref{thm:sth}}

We explain why \sthtwo-separation is in \pspace for monoids (as usual, the result may then be lifted to \nfas using Corollary~\ref{cor:autoreducvari}). The argument reuses the results of Section~\ref{sec:fixalph} and Appendix~\ref{app:fixalph}, and the fact that $\sthtwo = \pol{\at}$. In particular, we adapt Theorem~\ref{thm:efficient} to this setting. We start with some preliminary observations about the class \at.

\medskip

By definition of \at, it is straightforward to verify that the equivalence $\sim_\at$ compares words with the same alphabet. For $u,v \in A^*$, we have $u \sim_\at v$ if and only if $\cont{u} = \cont{v}$. Therefore, the monoid ${A^*}/{\sim_\at}$ corresponds to $2^A$ (the set of sub-alphabets) equipped with union as the multiplication. Moreover, for every $w \in A^*$, we have $\typ{\at}{w} = \cont{w}$.

We shall consider \at-compatible morphisms. If $\alpha: A^* \to M$ is \at-compatible, given $s \in M$, we shall write $\cont{s}$ for $\typ{\at}{s}$. We reuse the notion of $(\alpha,\beta,S)$-trees which we introduced in Section~\ref{sec:fixalph} (here, we use them in the special case when $\Cs = \at$).  Consider an alphabet $A$ and two \at-compatible morphisms $\alpha: A^* \to M$ and $\beta: A^* \to N$. Given a pair $(s,T) \in M \times 2^N$, we say that $(s,T)$ is alphabet safe when $\cont{s} = \cont{t}$ for every $t \in T$. The following lemma follows from definitions.

\begin{lemma} \label{lem:safe}

  Consider an alphabet $A$ and two \at-compatible morphisms $\alpha: A^* \to M$ and $\beta: A^* \to N$. Moreover, let $S \subseteq N$ be a good subset of $N$. Then, every $(s,T) \in M \times 2^N$ which is the root label of some $(\alpha,\beta,S)$-tree is alphabet safe.

\end{lemma}

Note that in the Appendix, the alphabet is one of our parameters which means that the size of the monoid ${A^*}/{\sim_\at} = 2^A$ may not be constant. Consequently, building \at-compatible morphisms is costly. Hence, we shall have to manipulate the construction explicitly. Given an arbitrary morphism $\alpha: A^* \to M$ into a finite monoid $M$, we write $\alpha_\at$ for the \at-compatible morphism $\alpha_\at: A^* \to M \times 2^A$ defined by $\alpha_\at(w) = (\alpha(w),\cont{w})$.

We may now adapt Theorem~\ref{thm:efficient} to this setting. This is the key result for proving that \sthtwo-separation is in \pspace for monoids.

\begin{proposition} \label{prop:keyprop}

  Consider two morphisms $\alpha: A^* \to M$ and $\beta: A^* \to N$. Moreover, let $\alpha_\at: A^* \to M \times 2^A$ and $\beta_\at: A^* \to N \times 2^A$ be the corresponding \at-compatible morphisms. Finally, let $S \subseteq N \times 2^A$ be a good subset of $N \times 2^A$ for $\beta_\at$.

  Given an alphabet safe pair $(s,T) \in (M \times 2^A) \times 2^{N \times 2^A}$, one may test in \pspace with respect to $|A|$, $|M|$ and $|N|$ whether there exists an $(\alpha_\at,\beta_\at,S)$-tree with root label $(s,T)$.

\end{proposition}

\begin{proof}[Proof sketch]

  By Lemma~\ref{lem:safe}, the set of possible labels for nodes in $(\alpha_\at,\beta_\at,S)$-trees has size at most $|M| \times 2^{|N|} \times 2^{|A|}$ (this is the size of the set of all alphabet safe pairs in $(M \times 2^A) \times 2^{N \times 2^A}$). This observation yields an \exptime least fixpoint algorithm for computing the set of all root labels of  $(\alpha_\at,\beta_\at,S)$-tree with root label $(s,T)$.

  This can be improved to \pspace by observing that it suffices to consider $(\alpha_\at,\beta_\at,S)$-trees whose heights are polynomially bounded with respect to $|A|$, $|M|$ and $|N|$. This is a simple consequence of Proposition~\ref{prop:opbound} since the \Js-depth of ${A^*}/{\sim_\at} = 2^A$ is easily verified to be $|A|+1$.

\end{proof}

Since $\sthtwo = \pol{\at}$, it is now simple to combine Theorem~\ref{thm:poltheo} with Proposition~\ref{prop:keyprop} to get a \pspace algorithm for \sthtwo-separation which concludes the proof.

\subsection{Proof of Lemma~\ref{lem:reduclem}}

Let us recall the statement of Lemma~\ref{lem:reduclem} (we refer the reader to Section~\ref{sec:classic} for the definition of the relevant notations).

\adjustc{lem:reduclem}

\begin{lemma}

  Consider $0 \leq i \leq n$. Then given an $i$-valuation $V$, the two following properties are equivalent:

  \begin{enumerate}

  \item $\Psi_i$ is satisfied by $V$.

  \item $L_i \cap [V]$ is not \sthtwo-separable from $L'_i \cap [V]$.

  \end{enumerate}

\end{lemma}

\restorec

We proceed by induction on $0 \leq i \leq n$. Let us start with the base case $i = 0$. In that case, $\Psi_0$ is the quantifier-free formula $\varphi$. Consider some $0$-valuation $V \subseteq (B_0)^*$. One may verify the following fact from the definitions of $L' \subseteq (B_0)^*$ and $[V]$.

\begin{fact} \label{fct:thefactreduc}

  The two following properties are equivalent:

  \begin{enumerate}

  \item $\Psi_0$ is satisfied by $V$.

  \item $L'_0  \cap [V] \neq \emptyset$.

  \end{enumerate}

\end{fact}

Since $L_0 = (B_0)^*$ by definition, we have $L_0 \cap [V] = [V]$. Hence, it is immediate that $L_0 \cap [V] = [V]$ is not \sthtwo-separable from $L'_0 \cap [V]$ if and only if $L'_0  \cap [V] \neq \emptyset$. Combined with Fact~\ref{fct:thefactreduc}, this yields Lemma~\ref{lem:reduclem} in the case $i=0$.

\medskip

We now assume that $i \geq 1$. There are two cases depending on whether the quantifier $Q_i$ is existential or universal (this is expected since the definitions of $L_i$ and $L'_i$ depend on this parameter). Since these two cases are similar, we handle the one when $Q_i$ is existential and leaver the other to the reader. Consider an $i$-valuation $V \subseteq (B_i)^*$. We have to show that the two following properties are equivalent:

\begin{enumerate}

\item $\Psi_i$ is satisfied by $V$.

\item $L_i \cap [V]$ is not \sthtwo-separable from $L'_i \cap [V]$.

\end{enumerate}

Let us start with some terminology that we shall use for both directions. We let $V_\bot$ and $V_\top$ as the following $(i-1)$-valuations built from $V$:
\[
  V_\top = V \setminus \{\#_i,\overline{x_i}\} \subseteq B_{i-1} \quad \text{and} \quad V_\bot = V \setminus \{\#_i,x_i\} \subseteq B_{i-1}
\]

We may now prove the equivalence. There are two directions to show.

\medskip

\noindent

{\bf Direction $1) \Rightarrow 2)$.} Assume that $\Psi_i$ is satisfied by $V$. We show that $L_i \cap [V]$ is not \sthtwo-separable from $L'_i \cap [V]$. We use Lemma~\ref{lem:sthsep}: given an arbitrary $k \in \nat$, we have to exhibit $w \in L_i \cap [V]$ and $w' \in L'_i \cap [V]$ such that $w \polrelk w'$. We fix $k$ for the proof.

Recall that by hypothesis, we have $\Psi_i = \exists x_i\ \Psi_{i-1}$. Hence, since $\Psi_i$ is satisfied by $V$, the definitions yield that either $V_\top$ or $V_\bot$ satisfies $\Psi_{i-1}$. By symmetry, we assume that we are in the former case: $V_\top$ satisfies $\Psi_{i-1}$. By induction hypothesis this implies that $L_{i-1} \cap [V_\top]$ is not \sthtwo-separable from $L'_{i-1} \cap [V_\top]$. Consequently, Lemma~\ref{lem:sthsep} yields $u \in L_{i-1} \cap [V_\top]$ and $u' \in L'_{i-1} \cap [V_\top]$ such that $u \polrelk u'$. Note that by definition of $V_\top$, we have $u,u' \in (B_{i-1} \setminus \{\overline{x_i}\})^*$. We define,
\[
  \begin{array}{lll}
    w  & = & (\#_i x_iu\$x_i)^{3^{k+1}}\#_i                                   \\
    y  & = & (\#_i x_iu\$x_i)^{3^{k+1}}\#_i\$(\#_i x_iu\$x_i)^{3^{k+1}}\#_i   \\
    w' & = & (\#_i x_iu'\$x_i)^{3^{k+1}}\#_i\$(\#_i x_iu'\$x_i)^{3^{k+1}}\#_i
  \end{array}
\]

Clearly, $\cont{\#_i\$} \subseteq \cont{\#_i x_iu\$x_i}$. Therefore, Lemma~\ref{lem:propreo2} yields that $w \polrelk y$. Moreover, since $u \polrelk u'$, we get from Lemma~\ref{lem:concat} that $y \polrelk w'$. By transitivity, we get $w \polrelk w'$. Finally, one may verify from the definition of $L_i$ and $L'_i$ that $w \in L_i \cap [V]$ and $w' \in L'_i \cap [V]$. Therefore, Lemma~\ref{lem:sthsep} yields that $L_i \cap [V]$ is not \sthtwo-separable from $L'_i \cap [V]$ as desired.

\medskip

\noindent

{\bf Direction $2) \Rightarrow 1)$.} We actually prove the contrapositive of this implication. Assuming that $\Psi_i$ is not satisfied by $V$, we show that $L_{i} \cap [V]$ is \sthtwo-separable from $L'_{i} \cap [V]$. Since $\Psi_i = \exists x_i\ \Psi_{i-1}$, our hypothesis yields that $\Psi_{i-1}$ is neither satisfied by $V_\top$ nor by $V_\bot$. Therefore, induction yields the two following properties:

\begin{enumerate}

\item $L_{i-1} \cap [V_\top]$ is \sthtwo-separable from $L'_{i-1} \cap [V_\top]$. We let $K_\top \in \sthtwo$ as a separator. Note that since $[V_\top] \in \sthtwo$ (actually $[V_\top] \in \at$), we may assume without loss of generality that $K_\top \subseteq [V_\top]$.

\item $L_{i-1} \cap [V_\bot]$ is \sthtwo-separable from $L'_{i-1} \cap [V_\bot]$.  We let $K_\bot \in \sthtwo$ as a separator. Again, we may assume without loss of generality that $K_\top \subseteq [V_\bot]$.

\end{enumerate}

We now define a language $K \in \sthtwo$ from $K_\top$ and $K_\bot$. We then show that it separates $L_{i} \cap [V]$ from $L'_{i} \cap [V]$. We let:
\[
  K =
  \begin{array}{ll}
    & \{\#_i\} \\
    \cup & A^*\#_i((A^*x_iA^* \cap A^*\overline{x_i}A^*) \setminus (A^*\#_iA^*))\#_i \\
    \cup & \#_ix_iK_{\top}\$x_i\#_i (A \setminus \{\overline{x_i}\})^* \\
    \cup & A^*\#_i((A^*\overline{x_i}A^*) \setminus (A^*\#_iA^*))\#_i x_iK_{\top}\$x_i\#_i (A \setminus \{\overline{x_i}\})^* \\
    \cup & \#_i\overline{x_i}K_{\bot}\$\overline{x_i}\#_i (A \setminus \{x_i\})^* \\
    \cup & A^*\#_i((A^*x_iA^*) \setminus (A^*\#_iA^*))\#_i \overline{x_i}K_{\bot}\$\overline{x_i}\#_i (A \setminus \{x_i\})^*
  \end{array}
\]

It is straightforward to verify that $K \in \pol{\at} = \sthtwo$. It remains to verify that $K$ separates $L_{i} \cap [V]$ from $L'_{i} \cap [V]$.

\medskip

We first show that $L_{i} \cap [V] \subseteq K$. Consider a word $w \in L_i \cap [V]$, we show that $w \in K$. Recall that we have $L_i = (\#_i(x_i + \overline{x_i})L_{i-1}\$(x_i + \overline{x_i}))^*\#_i$. Consequently, there exists $k \geq 0$  and $w_1,\dots,w_k \in  (x_i + \overline{x_i})L_{i-1}\$(x_i + \overline{x_i})$ such that,
\[
  w = \#_iw_1 \cdots \#_iw_k\#_i.
\]

Observe first that if $k = 0$, then $w = \#_i \in K$ and we are finished. Assume now that $k = 1$. By definition of $K$, when $w_k \in (A^*x_iA^* \cap A^*\overline{x_i}A^*) \setminus (A^*\#_iA^*)$, we also have $w \in K$. Therefore, we assume that $w_k \not\in (A^*x_iA^* \cap A^*\overline{x_i}A^*) \setminus (A^*\#_iA^*)$. Since $w_k \in (x_i + \overline{x_i})L_{i-1}\$(x_i + \overline{x_i})$, the letter $\#_i$ cannot occur in $w_k$ (by definition of $L_{i-1}$). Hence, our hypothesis on $w_k$ implies one of the two following properties holds:

\begin{itemize}

\item $x_i \in \cont{w_k}$ and $\overline{x_i} \not\in \cont{w_k}$, or,

\item $\overline{x_i} \in \cont{w_k}$ and $x_i \not\in \cont{w_k}$.

\end{itemize}

By symmetry, we handle the case when the first property holds and leave the other to the reader. We now assume that $x_i \in \cont{w_k}$ and $\overline{x_i} \not\in \cont{w_k}$.

There are two sub-cases depending on whether $\overline{x_i} \in \cont{w}$ or not. Assume first that $\overline{x_i} \not\in \cont{w}$. Since $w_1 \in (x_i + \overline{x_i})L_{i-1}\$(x_i + \overline{x_i})$, it follows that $w_1 = x_i u\$x_i$ where $u \in L_{i-1}$. Moreover, recall that $w \in [V]$ by definition which implies that $u \in [V]$. Moreover, $\cont{u}$ contains neither $\overline{x_i}$ nor $\#_i$ (the latter holds by definition of $L_{i-1}$). Altogether, this yields that $u \in L_{i-1} \cap [V_\top]$ and therefore $u \in K_\top$ by definition of $K_\top$. It follows that $w_1 \in x_iK_{\top}\$x_i$ which implies that $w \in \#_ix_iK_{\top}\$x_i\#_i (A \setminus \{\overline{x_i}\})^* \subseteq K$ which concludes this case.

Finally, assume that $\overline{x_i} \in \cont{w}$. Therefore, there exists some factor $w_j$ for $j \leq k$  such that $\overline{x_i} \in \cont{w_j}$. We consider the rightmost one. Note that we have $j < k$ by hypothesis on $w_k$. By definition, we know that $\overline{x_i} \not\in \cont{\#_iw_{j+1} \cdots \#_iw_k\#_i}$. We may now reuse the argument of the previous case to obtain that,
\[
  \#_iw_{j+1} \cdots \#_iw_k\#_i \in \#_ix_iK_{\top}\$x_i\#_i (A \setminus \{\overline{x_i}\})^*
\]

Moreover, by definition of $w_j$, we have $w_j \in (A^*\overline{x_i}A^*) \setminus (A^*\#_iA^*)$. Therefore, we obtain,
\[
  w \in A^*\#_i((A^*\overline{x_i}A^*) \setminus (A^*\#_iA^*))\#_i x_iK_{\top}\$x_i\#_i (A \setminus \{\overline{x_i}\})^* \subseteq K
\]

This concludes the proof that $L_i \subseteq K$.

\medskip

It remains to show that $L'_i \cap [V] \cap K = \emptyset$. We proceed by contradiction and assume that there exists $w \in L'_i \cap [V] \cap K$. Recall that by definition, we have
\[
  \begin{array}{c}
    T_i = (\#_ix_i (B_{i-1} \setminus \{\overline{x_i}\})\$x_i)^* \quad \text{and} \quad  \overline{T_i} = (\#_i \overline{x_i} (B_{i-1} \setminus \{x_i\})\$\overline{x_i})^* \\
    L'_i  = (\#_i (x_i + \overline{x_i})L'_{i-1}\$(x_i + \overline{x_i}))^* \#_i\$
    \left(
    T_i\#_i \cup \overline{T_i}\#_i
    \right)
  \end{array}
\]

Therefore, since $w \in L'_i$, we have $w = u\#_i\$ v\#_i$ with $u \in (\#_i (x_i + \overline{x_i})L'_{i-1}\$(x_i + \overline{x_i}))^*$ and $v \in T_i \cup \overline{T_i}$. By symmetry, we shall assume that $v \in T_i$. We obtain that $k,\ell \geq 0$  and $u_1,\dots,u_k \in  (x_i + \overline{x_i})L'_{i-1}\$(x_i + \overline{x_i})$ and $v_1,\dots,v_\ell \in \#_ix_i (B_{i-1} \setminus \{\overline{x_i}\})\$x_i$ such that,
\[
  u = \#_iu_1 \cdots \#_iu_k \quad \text{and} \quad v = \#_iv_1 \cdots \#_iv_\ell
\]

Since $K$ is defined as a union, $w$ belongs to some member of this union. We treat each case independently. If $w \in \{\#_i\}$, we have a contradiction since $w$ contains the letter $\$$ by definition.

Assume now that $w \in A^*\#_i((A^*x_iA^* \cap A^*\overline{x_i}A^*) \setminus (A^*\#_iA^*))\#_i$. If $\ell = 0$, this means that $\$ \in (A^*x_iA^* \cap A^*\overline{x_i}A^*) \setminus (A^*\#_iA^*)$ which is a contradiction. Otherwise $\ell \geq 1$ and we obtain that $v_\ell \in (A^*x_iA^* \cap A^*\overline{x_i}A^*) \setminus (A^*\#_iA^*)$. This is also a contradiction since $v_\ell \in \#_ix_i (B_{i-1} \setminus \{\overline{x_i}\})\$x_i$ and cannot contain the letter $\overline{x_i}$.

We now treat the case when $w \in \#_ix_iK_{\top}\$x_i\#_i (A \setminus \{\overline{x_i}\})^*$. If $k = 0$, this implies that $\$ \in x_iK_{\top}\$x_i$ which is a contradiction. Otherwise, we have $u_1 \in x_iK_{\top}\$x_i$. Recall that $u_1 \in  (x_i + \overline{x_i})L_{i-1}\$(x_i + \overline{x_i})$. Therefore, $u_1 \in x_i L'_{i-1}\$ x_i$ which implies that $L'_{i-1} \cap K_{\top} \neq \emptyset$. Furthermore, since $K_\top \subseteq [V_\top]$ by definition, we get that $L'_{i-1} \cap [V_\top] \cap K_{\top} \neq \emptyset$. This contradicts the definition of $K_\top$. One may handle the case when $w \in \#_i\overline{x_i}K_{\bot}\$\overline{x_i}\#_i (A \setminus \{x_i\})^*$ symmetrically using the definition of $K_\bot$.

We turn to the case when $w \in A^*\#_i((A^*\overline{x_i}A^*) \setminus (A^*\#_iA^*))\#_i x_iK_{\top}\$x_i\#_i (A \setminus \{\overline{x_i}\})^*$. Since the factors $v_j$ cannot contain the letter $\overline{x_i}$, it follows that there exists $j \leq k$ such that $u_j \in (A^*\overline{x_i}A^*) \setminus (A^*\#_iA^*)$ and,
\[
  \#_iu_{j+1} \cdots \#_iu_k\#_i\$v\#_i \in \#_i x_iK_{\top}\$x_i\#_i (A \setminus \{\overline{x_i}\})^*
\]

One may now reuse the argument of the previous case to derive a contradiction. Finally, one may handle that case when $w \in A^*\#_i((A^*x_iA^*) \setminus (A^*\#_iA^*))\#_i \overline{x_i}K_{\bot}\$\overline{x_i}\#_i (A \setminus \{x_i\})^*$ symmetrically which concludes the proof.

\subsection{Proof of Theorem~\ref{thm:bpoltheo}}

It is straightforward to verify from Proposition~\ref{prop:keyprop} and Theorem~\ref{thm:bpoltheo} that \sttwo-separation is in \exptime for monoids (since \sttwo is a \varie, this is also the case for \nfas by Corollary~\ref{cor:autoreducvari}). We focus on proving that \sttwo-separation is \pspace-hard for \nfas (again this is lifted to monoids with Corollary~\ref{cor:autoreducvari}). As explained in the main paper, this boils down to proving Proposition~\ref{prop:bpolred}.

\adjustc{prop:bpolred}

\begin{proposition}

  Consider an alphabet $A$ and $H,H' \subseteq A^*$. Let $B = A \cup \{\#,\$\}$ with $\#,\$ \not\in A$, $L = \#(H'\#(A^*\$\#)^*)^*H\#(A^*\$\#)^* \subseteq B^*$ and $L' = \#(H'\#(A^*\$\#)^*)^* \subseteq B^*$. The two following properties are equivalent:

  \begin{enumerate}

  \item $H$ is \sthtwo-separable from $H'$.

  \item $L$ is \sttwo-separable from $L'$.

  \end{enumerate}

\end{proposition}

\restorec

We start with the direction $1) \Rightarrow 2)$. Assume that $H$ is \sthtwo-separable from $H'$ and let $K \subseteq A^*$ be a separator in \sthtwo. Consider the following language $S \subseteq B^*$:
\[
  S = B^*\#K\#B^*.
\]

Clearly, $S \in \sthtwo \subseteq \sttwo$. Moreover, since $L = \#(H'\#(A^*\$\#)^*)^*H\#(A^*\$\#)^*$ and $H \subseteq K$ by definition of $K$, we have $L \subseteq S$. Finally, we have $H' \cap K = \emptyset$ by definition of $K$. Moreover, $L' = \#(H'\#(A^*\$\#)^*)^*$. Since $\#,\$ \not\in A$, given $w \in L'$, the only factors of $w$ belonging to $\#A^*\#$ actually belong to $\#H'\#$.  Therefore, since $K \subseteq A^*$, we get $L' \cap K = \emptyset$ which concludes the proof for the direction $1) \Rightarrow 2)$.

\medskip

We turn to the direction $2) \Rightarrow 1)$. Actually, we prove the contrapositive. Assuming that $H$ is not \sthtwo-separable from $H'$, we show that $L$ is not \sttwo-separable from $L'$. By Lemma~\ref{lem:stsep}, we have to show that for every $k \in \nat$, there exists $w \in L$ and $w' \in L'$ such that $w \polrelk w'$ and $w' \polrelk w$. we fix $k$ for the proof.

Since $H$ is not \sthtwo-separable from $H'$, Lemma~\ref{lem:stsep} yields $u \in  H$ and $u' \in H'$ such that $u \polrelk u'$. We define,
\[
  \begin{array}{lll}
    w  & = & \#(u'\#(u\$\#)^{3^{k+1}})^{3^{k+1}}u\#(u\$\#)^{3^{k+1}}   \\
    w' & = & \#(u'\#(u\$\#)^{3^{k+1}})^{3^{k+1}}
  \end{array}
\]

Since $u \in H$ and $u' \in H'$, it is clear from the definitions of $L$ and $L'$ that $w \in L$ and $w' \in L'$. It remains to show that $w \polrelk w'$ and $w' \polrelk w$. We start with the former.

Since $u \polrelk u'$, we may use Lemma~\ref{lem:concat} to obtain the following inequality:
\[
  w \polrelk \#(u'\#(u\$\#)^{3^{k+1}})^{3^{k+1}}u'\#(u\$\#)^{3^{k+1}} = \#(u'\#(u\$\#)^{3^{k+1}})^{3^{k+1}+1}
\]

Moreover, it is immediate from Lemma~\ref{lem:propreo1} that we have,
\[
  \#(u'\#(u\$\#)^{3^{k+1}})^{3^{k+1}+1} \polrelk w'
\]

By transitivity, this yields $w \polrelk w'$.

\medskip
We finish with the converse inequality. Clearly, $\cont{u\#} \subseteq \cont{u\$\#}$. Therefore,  Lemma~\ref{lem:propreo2} yields that,
\[
  (u\$\#)^{3^{k+1}} \polrelk (u\$\#)^{3^{k+1}}u\#(u\$\#)^{3^{k+1}}
\]

We may apply Lemma~\ref{lem:concat} to obtain:
\[
  \#(u'\#(u\$\#)^{3^{k+1}})^{3^{k+1}} \polrelk \#(u'\#(u\$\#)^{3^{k+1}})^{3^{k+1}}u\#(u\$\#)^{3^{k+1}}
\]

This exactly says that $w' \polrelk w$, finishing the proof.
\end{document}